\documentclass[a4paper,cleveref, autoref, thm-restate]{lipics-v2021}

\pdfoutput=1 %uncomment to ensure pdflatex processing (mandatatory e.g. to submit to arXiv)
\hideLIPIcs  %uncomment to remove references to LIPIcs series (logo, DOI, ...), e.g. when preparing a pre-final version to be uploaded to arXiv or another public repository

\title{Bidding Games with Charging} %TODO Please add

\author{Guy Avni}{University of Haifa, Israel \and \url{https://sites.google.com/view/gavni} }{gavni@cs.haifa.ac.il}{https://orcid.org/0000-0002-1825-0097}{}%{(Optional) author-specific funding acknowledgements}%TODO mandatory, please use full name; only 1 author per \author macro; first two parameters are mandatory, other parameters can be empty. Please provide at least the name of the affiliation and the country. The full address is optional. Use additional curly braces to indicate the correct name splitting when the last name consists of multiple name parts.
\author{Ehsan Kafshdar Goharshady}{Institute of Science and Technology Austria (ISTA), Austria \and \url{https://ehsan.goharshady.com/} }{ehsan.goharshady@ist.ac.at}{https://orcid.org/0000-0002-8595-0587}{}%{(Optional) author-specific funding acknowledgements}
\author{Thomas A.\ Henzinger}{Institute of Science and Technology Austria (ISTA), Austria \and \url{https://pub.ista.ac.at/~tah/} }{tah@ist.ac.at}{https://orcid.org/0000-0002-2985-7724}{}
\author{Kaushik Mallik}{Institute of Science and Technology Austria (ISTA), Austria \and \url{https://kmallik.github.io/} }{kaushik.mallik@ist.ac.at}{https://orcid.org/0000-0001-9864-7475}{}
\authorrunning{Avni et al.} %TODO mandatory. First: Use abbreviated first/middle names. Second (only in severe cases): Use first author plus 'et al.'
\Copyright{Guy Avni, Ehsan Kafshdar Goharshady, Thomas A.\ Henzinger, and Kaushik Mallik} %TODO mandatory, please use full first names. LIPIcs license is "CC-BY";  http://creativecommons.org/licenses/by/3.0/
\ccsdesc[500]{Theory of computation~Algorithmic game theory}
\keywords{Bidding games on graphs, $\omega$-regular objectives} %TODO mandatory; please add comma-separated list of keywords
\funding{This work was supported in part by the ERC projects ERC-2020-AdG 101020093 and CoG 863818 (ForM-SMArt)
and by ISF grant no. 1679/21.}%optional, to capture a funding statement, which applies to all authors. Please enter author specific funding statements as fifth argument of the \author macro.
\nolinenumbers %uncomment to disable line numbering
\usepackage{amsmath,amsthm}
\usepackage{mathtools}
\usepackage{xspace}
\usepackage{xfrac}
\usepackage{algorithm,algpseudocode}
\usepackage{tikz,pgfplots}
\usetikzlibrary{automata, arrows.meta, positioning}
\usepackage{caption,subcaption}
\usepackage{wrapfig}
\usepackage{multirow}
\usepackage{colortbl}
\usepackage{enumerate}
\usepackage{diagbox}
\usepackage{wrapfig}
\usepackage[framemethod=tikz]{mdframed}
\usepackage{graphicx,caption}

\mdfdefinestyle{mystyle}{%
  rightline=true,
  innerleftmargin=10,
  innerrightmargin=10,
  outerlinewidth=3pt,
  topline=false,
  rightline=true,
  bottomline=false,
  skipabove=\topsep,
  skipbelow=\topsep
}

\newcommand{\tup}[1]{\left\langle #1\right\rangle}
\newcommand{\set}[1]{\left\lbrace #1\right\rbrace}
\newcommand{\clamp}[1]{\mathtt{clamp}_{[0,1]}\left(#1\right)}
\newcommand{\compl}[1]{\overline{#1}}

\providecommand{\G}{}
\renewcommand{\G}{\mathcal{G}}
\newcommand{\V}{V}

\newcommand{\E}{E}

\newcommand{\RT}{R_2}
\newcommand{\RO}{R_1}
\newcommand{\RI}{R_i}
\newcommand{\BO}{B_1}
\newcommand{\BT}{B_2}

\newcommand{\bo}{b_1}
\newcommand{\bt}{b_2}
\newcommand{\spec}{\varphi}
\newcommand{\PO}{Player~1\xspace}
\newcommand{\N}{\mathcal{S}}
\newcommand{\PT}{Player~2\xspace}

\newcommand{\strat}{\pi}

\newcommand{\stratI}{\pi_i}
\newcommand{\play}{\mathsf{play}}
\newcommand{\Th}{\mathit{Th}}
\newcommand{\ThO}{\Th_1}
\newcommand{\ThT}{\Th_2}

\newcommand{\bThO}{g^*_1}
\newcommand{\bThT}{g^*_2}

\newcommand{\bthO}{g_1}
\newcommand{\bthT}{g_2}

\newcommand{\upd}{\mathit{Av}}
\newcommand{\updO}{\upd_1}
\newcommand{\updT}{\upd_2}
\newcommand{\updI}{\upd_i}

\newcommand{\THRESH}{\mathrm{THRESH}}

\mathchardef\mhyphen="2D % Define a "math hyphen"
\newcommand{\safe}{\mathrm{Safe}}
\newcommand{\reach}{\mathrm{Reach}}

\newcommand{\buchi}{\mathrm{B\ddot{u}chi}}
\newcommand{\cobuchi}{\mathrm{Co \mhyphen B\ddot{u}chi}}
\newcommand{\frugalreach}{\mathrm{FrugalReach}}
\newcommand{\rabin}{\mathrm{Rabin}}
\newcommand{\streett}{\mathrm{Streett}}
\newcommand{\f}{\mathsf{fr}}

\newcommand{\stam}[1]{}
\newcommand{\zug}[1]{\langle #1 \rangle}
\renewcommand{\set}[1]{\{ #1  \}}

\newcommand{\PLi}{Player~$i$\xspace}

\newcommand{\Real}{\mathbb{R}}
\newcommand{\Nat}{\mathbb{N}}

\renewcommand{\qed}{\hfill \ensuremath{\Box}}

\begin{document}

\maketitle

\begin{abstract}
\emph{Graph games} lie at the algorithmic core of many automated design problems in computer science.
These are games usually played between two players on a given graph, where the players keep moving a token along the edges according to pre-determined rules (turn-based, concurrent, etc.), and the winner is decided based on the infinite path (aka \emph{play}) traversed by the token from a given initial position. 
In bidding games, the players initially get some monetary budgets which they need to use to bid for the privilege of moving the token  at each step.
Each round of bidding affects the players' available budgets, which is the only form of update that the budgets experience.
We introduce bidding games {\em with charging} where the players can additionally improve their budgets during the game by collecting vertex-dependent monetary rewards, aka the ``charges.''
%At each step, after the token moves to a new location, the respective charges get added to the players' budgets.
Unlike traditional bidding games (where all charges are zero), bidding games with charging allow non-trivial recurrent behaviors.
For example, a reachability objective may require multiple detours to vertices with high charges to earn additional budget.
We show that, nonetheless, the central property of traditional bidding games generalizes to bidding games with charging: For each vertex there exists a \emph{threshold} ratio, which is the necessary and sufficient fraction of the total budget for winning the game from that vertex.
While the thresholds of traditional bidding games correspond to unique fixed points of linear systems of equations, in games with charging, these fixed points are no longer unique. 
This significantly complicates the proof of existence and the algorithmic computation of thresholds for infinite-duration objectives.
We also provide the lower complexity bounds for computing thresholds for Rabin and Streett objectives, which are the first known lower bounds in any form of bidding games (with or without charging), and we solve the following \emph{repair problem} for safety and reachability games that have unsatisfiable objectives:
Can we distribute a given amount of charge to the players in a way such that the objective can be satisfied?
\end{abstract}

\section{Introduction}
Two-player {\em graph games} have deep connections to foundations of mathematical logic \cite{Rab69}, and constitute a fundamental model of computations with applications in {\em reactive synthesis}~\cite{PR89} and multi-agent systems~\cite{AHK02}.
A graph game is played on a graph, called the \emph{arena}, as follows. A token is placed on an initial vertex and the two players move the token throughout the arena to produce an infinite path, called a {\em play}.
The winner is determined based on whether the play fulfills a given temporal objective (or specification). Traditionally, graph games are {\em turn-based}, where the players move the token in alternate turns. 
{\em Bidding games} are graph games where who moves the token at each step is determined by an auction (a {\em bidding}). Concretely, both players are allocated initial budgets, and in each turn, they concurrently place bids from their available budgets, the highest bidder moves the token, and pays his bid according to one of the following pre-determined mechanisms. 
In {\em Richman} bidding, the bid is paid to the lower bidder, in {\em poorman} bidding, the bid is paid to an imaginary ``bank'' and the money is lost, and in {\em taxman} bidding, a fixed fraction of the bid is paid to the bank (the ``tax'') and the rest goes to the lower bidder.
The outcome of the game is an infinite play and, as usual, the winner is determined based on whether the play fulfills a given objective.

Bidding games model strategic decision-making problems where resources need to be invested dynamically towards the fulfillment of an objective. 
For example, a taxi driver needs to decide how to ``invest'' his gas supply in order to collect as many passengers as possible, internet advertisers need to invest their advertising budgets in ongoing auctions for advertising slots with the goal of maximizing visibility~\cite{AHI18}, or a coach in an NBA tournament needs to decide his roster for each game while ``investing'' his players' limited energy with the goal of winning the tournament~\cite{AIT20}. 
While in all these scenarios the investment resources can be ``charged,'' e.g., by visiting a gas station, by adding funds, or by allowing the players to rest, respectively, charging budgets cannot be modeled in traditional bidding games. 

We study, for the first time, \emph{bidding games with charging}, where the players can increase their available budgets by collecting vertex-dependent charges.
Every vertex $v$ in the arena is labeled with a pair of non-negative rational numbers denoted $R_1(v) $ and $ R_2(v)$.  
Suppose the game enters a vertex $v$, where for $i \in \set{1,2}$, \PLi's budget is $B_i$ with $B_1 + B_2 = 1$. 
First, the budgets are charged to $B'_1 = B_1 +R_1(v)$ and $B'_2 = B_2 + R_2(v)$. Second, we normalize the sum of budgets to $1$ by defining $B''_1 = B'_1/(B'_1 + B'_2)$ and $B''_2 = B'_2/(B'_1 + B'_2)$. Finally, the players bid from their new available budgets $B''_1$ and $B''_2$, and the bids are resolved using any of the traditional mechanisms. Note that traditional bidding games are a special case of bidding games with charging in which all charges are $0$. 
The normalization step plays an important role and will be discussed in Ex.~\ref{ex:normalization}.

\begin{figure}
	\centering
	\tikzset{every state/.style={minimum size=0pt}}
	\begin{subfigure}[b]{0.45\textwidth}
		\centering
		\begin{tikzpicture}[node distance=0.6cm]
			\node[state,label={below:$\begin{bmatrix}2\\ 0\end{bmatrix}$},label={above left:{\color{blue}$\mathtt{0}$}}]	(a)	at	(0,0)		{$a$};
			\node[state,label={above:{\color{blue}$\mathtt{\frac{1}{4}}$}}]	(b)	[right=of a]	{$b$};
			\node[state,label={above:{\color{blue}$\mathtt{\frac{1}{2}}$}}]	(c)	[right=of b]	{$c$};
			\node[state,accepting,label={above left:{\color{blue}$\mathtt{0}$}}]	(d)	[above right=of c]	{$d$};
			\node[state,label={left:{\color{blue}$\mathtt{1}$}}]	(e)	[below right=of c]	{$e$};
			
			\path[->]	(a)	edge[bend left]	(b)
							edge[loop above]	()
						(b)	edge[bend left]	(a)
							edge			(c)
						(c)	edge			(d)
							edge			(e)
						(d)	edge[loop above]()
						(e)	edge[loop above]	();
						
		\end{tikzpicture}
		\caption{Strategies may depend on the available budget.}
		\label{fig:examples:a}
	\end{subfigure}
	\quad
	\begin{subfigure}[b]{0.45\textwidth}
		\centering
		\begin{tikzpicture}[node distance=0.6cm]
			\node[state,label={left:$\begin{bmatrix}0\\ 6\end{bmatrix}$},label={right:{\color{blue}$\mathtt{1}$}}]	(a)	at	(0,0)		{$a$};
			\node[state,accepting,label={right:{\color{blue}$\mathtt{0}$}}]	(t)	[above right=of a]	{$t$};
			\node[state,label={left:$\begin{bmatrix}0.25\\ 0\end{bmatrix}$},label={right:{\color{blue}$\mathtt{\frac{3}{8}}$}}]	(b)	[above left=of t]	{$b$};

			\path[->]	(a)	edge[bend left]	(b)
							edge	(t)
						(b)	edge[bend left]	(a)
							edge	(t);
						
		\end{tikzpicture}
		\caption{Nontrivial solution for safety objective of \PT when the unsafe vertex (which is $t$) is reachable from every other vertex.}
		\label{fig:examples:c}
	\end{subfigure}
	\caption{Examples to demonstrate the distinctive features of bidding games with charging, compared to traditional bidding games (without charging).
	The double circled vertices are the ones that \PO wants to reach (reachability objective), or, dually, the ones that \PT wants to avoid (safety objective).
	When a vertex $v$ has nonzero reward for at least one of the players, the rewards are shown next to $v$ in the vector notation 
	$\begin{bmatrix}
		\RO(v)\\
		\RT(v)	
	\end{bmatrix}$.
	The threshold budget of \PO for each vertex is shown in {\color{blue}\texttt{blue}} next to the vertex.}
	\label{fig:examples}
\end{figure}
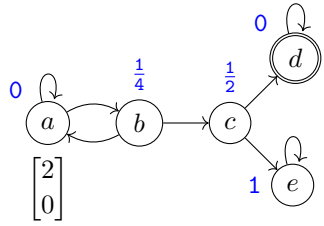
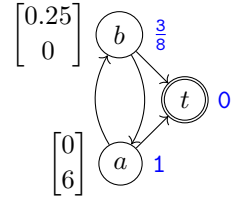

\begin{example}
\label{ex:bud-agnostic}%[Strategies may not be budget-agnostic]
We illustrate the model and show a distinctive feature that is not present in traditional bidding games. Consider the bidding game in Fig.~\ref{fig:examples:a}, where the objective of \PO, the reachability player, is to reach $d$ and the objective of \PT, the safety player, is to prevent this. Consider Richman bidding. 
We show that from vertex $b$, \PO can win with a budget of $B_1 = \frac{1}{4} + \epsilon \leq 1$ for every $\epsilon > 0$. 
\PO bids $\frac{1}{4}$ at $b$. We consider two cases. 
\emph{First}, \PT wins the bidding and proceeds to $c$. She pays \PO at least $\frac{1}{4}$, and \PO's budget at $c$ becomes at least $\frac{1}{2}+\epsilon$.
\PO can now win the bidding by bidding all of his budget (recall that the sum of budgets is $1$), and can proceed to $d$ to win the game. 
\emph{Second}, suppose that \PT loses the bidding at $b$. \PO proceeds to $a$ with a budget of $\epsilon$. We charge his budget to $B'_1 = 2+\epsilon$ and after re-normalizing his new budget becomes $B''_1 = \frac{B'_1}{3} > \frac{2}{3}$. 
\PO increases his budget by forcing the game to stay in $a$ for three consecutive turns: He first bids $\frac{1}{3}$, and his budget exceeds $(\frac{1}{3}+2)/3 = \frac{7}{9}$, then he bids $\frac{2}{9}$ and $\frac{4}{27}$ in the following two turns, after which his budget exceeds $\frac{73}{81} > \frac{7}{8}$. Since every budget greater than $\frac{7}{8}$ suffices to guarantee winning three consecutive biddings, he can now force the game to reach $d$, resulting in a win. 

We point out a distinction from traditional bidding games (without charging). 
In games without charging, it is known that if a player wins, he can win using a  {\em budget agnostic} winning strategy: For every vertex $v$, there is a successor $u$ such that upon winning the bidding at $v$, the strategy proceeds to $u$ regardless of the current available budget.\footnote{We refrain from calling the strategy {\em memoryless} since it might bid differently in successive visits to $v$.} 
However, it is not hard to see that there is no winning budget-agnostic strategy in the game above; indeed, in order to win, \PO must eventually go right from $b$, but when his budget is $0.25 < B_1 \leq 0.75$, he needs to go left and going right will make him lose. \qed
\end{example}

Another distinctive feature of bidding games with charging is that safety games have non-trivial solutions, while in traditional bidding games, the only way to ensure safety is by reaching a vertex with no path to the unsafe vertices~\cite{LLPSU99,AHC19,AHZ21}. 
Therefore, charging opens doors to new applications of bidding games for when safety objectives are involved.
For example, in {\em auction-based scheduling}~\cite{AMS24}, bidding games are used to compose two policies at runtime such that the objectives of both policies are fulfilled. With traditional bidding games, auction-based scheduling cannot support long-run safety due to the aforementioned reasons.
Bidding games with charging creates the possibility to extend auction-based scheduling for richer classes of objectives than what can be supported currently.

\begin{example}
	\label{ex:non-trivial-thresholds}
	We show that \PT, the safety player, wins the game depicted in Fig.~\ref{fig:examples:c} starting from $b$ when \PO's budget is $B_1 < \frac{3}{8}$. 
	Fulfilling safety requires the game to forever loop over $a$ and $b$; such an outcome is not possible in traditional bidding games since $t$ is reachable from both $a$ and $b$. 
	After charging at $b$, we have $B''_1 < \frac{1}{2}$. \PT bids $\frac{1}{2}$, trivially wins the bid and moves the token to $a$. Her budget is charged to at least $\frac{6}{7}$, meaning that \PO's budget is at most $\frac{1}{7}$. She bids $\frac{3}{16}$, trivially wins the bidding and move the token to $b$. When entering $b$ her budget is at least $\frac{6}{7}-\frac{3}{16}>\frac{5}{8}$, meaning that \PO's budget is less than $\frac{3}{8}$, and she can keep repeating the same strategy to win the game. \qed
\end{example}

The central quantity in bidding games is the pair of {\em thresholds} on the players' budgets which enable them to win.
Formally, for $i \in \set{1,2}$, \PLi's threshold at vertex $v$, denoted $\Th_i(v)$, is the smallest value in $[0,1]$ such that for every $\epsilon >0$, \PLi can guarantee winning from $v$ with an initial budget of $\Th_i(v) + \epsilon$. The thresholds in the vertices in Figures~\ref{fig:examples:a} and~\ref{fig:examples:c} are depicted beside them in blue. 
When $\Th_1(v) + \Th_2(v) = 1$, we say that a {\em threshold} exists and define the threshold to be $\Th(v) = \Th_1(v)$. Existence of thresholds is a form of {\em determinacy}: for every \PO budget $B_1 \neq \Th_1(v)$, one of the players has a winning strategy. 
We establish that bidding games with charging are also determined for reachability and B\"uchi objectives, and, dually, for safety and co-B\"uchi objectives. 
The proofs of these claims are however significantly more involved than the case of traditional bidding games. 
For instance, for traditional bidding games, the existence of thresholds for B\"uchi objectives follows from the existence of thresholds for reachability objectives, with the observation that for every bottom strongly connected component (BSCC), every vertex has a threshold $0$ or $1$, so that winning the B\"uchi game boils down to \emph{reaching} one of the BSCCs with thresholds $0$ (the ``winning'' BSCCs).
This approach fails for games with charging. First, players may be able to trap the game within an SCC that is not part of any BSCC, and second, the thresholds in a BSCC might not be all $0$ or $1$ as seen in Ex.~\ref{ex:non-trivial-thresholds}. In order to show the existence of thresholds in B\"uchi games, we develop a novel fixed point algorithm that is based on repeated solutions to reachability bidding games. 

We study the complexity of finding thresholds. Here too, the techniques differ and are more involved than traditional bidding games. 
In Richman-bidding games without charging, thresholds correspond to the unique solution of a system of linear equations. In games with charging, however, thresholds correspond to the least and the greatest fixed points, and we present a novel encoding of the problem using mixed-integer linear programming. 
We summarize our complexity results in Tab.~\ref{table:complexity upper-bounds} along with a comparison with known results in traditional bidding games.
Finally, we show that Richman games with Rabin and Streett objectives are NP-hard and coNP-hard, respectively. 
This result establishes the first lower complexity bound in any form of bidding games (with or without charging).
Upper bounds for Rabin and Streett objectives are left open.

Finally, we introduce and study a {\em repair} problem in bidding games: Given a bidding game, a target threshold $t$ in a vertex $v$, and a repair budget $C$, decide if it is possible to add charges to the vertices of $\G$ in a total sum that does not exceed $C$ such that $\Th(v) \leq t$. 
Repairing is relevant when the bidding game is not merely given to us as a fixed input, but rather the design of the game is part of the solution itself.
For instance, we have already mentioned auction-based scheduling~\cite{AMS24}, where the strongest guarantees can be provided when in two bidding games that are played on the same arena, the sum of thresholds in the initial vertex is less than $1$. When this requirement fails, repairing can be applied to lower the thresholds. 
We show that the repair problem for safety objectives is in PSPACE and for reachability objectives is in 2EXPTIME.

\newcommand{\woch}{w/o chg.}
\newcommand{\ch}{w/ chg.}
\colorlet{MyYellow}{yellow!50!white}
\newcolumntype{a}[1]{>{\centering\columncolor{MyYellow}}m{#1}}
\begin{table*}
	\centering
%	\begin{tabular}{|m{1.2cm}|a|c|a|c|a|c|a|c|}
	\begin{tabular}{|m{1.2cm}|a{0.8cm}|c|a{0.8cm}|c|a{0.8cm}|c|a{0.8cm}|c|}
		\hline
		\multirow{2}{0.5cm}{}&	\multicolumn{2}{c|}{Reachability}	&	\multicolumn{2}{c|}{Safety}	&	\multicolumn{2}{c|}{B\"uchi}	&	\multicolumn{2}{c|}{Co-B\"uchi}\\
		\cline{2-9}
		 &	\ch	& \woch &	\ch	& \woch &	\ch	& \woch &	\ch & \woch	\\
		\hline
		\small Richman	& \small	\tiny coNP	&	\tiny $\textup{NP}\cap\textup{coNP}$	&	\tiny NP	&	\tiny $\textup{NP}\cap\textup{coNP}$	&	\tiny $\Pi_2^\mathrm{P}$	&	\tiny $\textup{NP}\cap\textup{coNP}$	&	\tiny $\Sigma_2^\mathrm{P}$	&	\tiny $\textup{NP}\cap\textup{coNP}$\\
		\hline
		\small Taxman and poorman	&	\tiny PSPACE	&	\tiny PSPACE	&	\tiny PSPACE	&	\tiny PSPACE	&	\tiny 2-EXP	&	\tiny PSPACE	&	\tiny 2-EXP	&	\tiny PSPACE\\
		\hline	
	\end{tabular}
	\caption{Upper complexity bounds for bidding games with charging (``w/ chg.'') in comparison with traditional bidding games (``w/o chg.'').}
	\label{table:complexity upper-bounds}
\end{table*}

\subsection*{Related work}
Bidding games (without charging) were introduced by Lazarus et al.~\cite{LLPU96,LLPSU99}, and were extended to infinite-duration objectives by Avni et al.~\cite{AHC19,AHI18,AHZ19,AJZ21}.
Many variants of bidding games have been studied, including {\em discrete}-bidding~\cite{DP10,AAH21,AS22}, which restricts the granularity of bids, {\em all-pay} bidding~\cite{AIT20,AJZ21}, which model allocation of non-refundable resources, partial-information games~\cite{AJZ23}, which restricts the observation power of one of the players, and non-zero-sum bidding games~\cite{MKT18}, which allow the players' objectives to be non-complimentary.
The inspiration behind charging comes from other forms of resource-constrained games that allow to refill depleted resources or accumulate new resources to perform certain tasks \cite{chakrabarti2003resource,bouyer2008infinite,blahoudek2022efficient}.
The unique challenge in our case is the additional layer of bidding, which separates resource (budget) accumulation and spending.

%!TEX root=main.tex

\section{Bidding Games with Charging}
\label{sec:model}

A \emph{bidding game with charging} is a two-player game played on an arena $\tup{\V,\E,\RO,\RT}$ between \PO and \PT,\footnote{We will use the pronouns ``he'' and ``she'' for \PO and \PT, respectively.} where 
$\V$ is a finite set of vertices, 
$\E\subseteq \V\times\V$ is a set of directed edges, 
and $\RO,\RT: V \rightarrow \Real_{\geq 0}$ are the charging functions of \PO and \PT, respectively. 
We denote the set of {\em successors} of the vertex $v$ by $\N(v) = \set{u: \zug{v,u} \in E}$.
Bidding games with no charging will be referred to as \emph{traditional} bidding games, which is a special case with $\RO \equiv \RT \equiv 0$. 
The default ones in this paper are bidding games with charging and, to avoid clutter, we typically refer to them simply as {\em bidding games}.

A bidding game proceeds as follows. 
A {\em configuration} of a bidding game is a pair $c = \zug{v, B_1} \in V \times [0,1]$, which indicates that the token is placed on the vertex $v$ and \PO's current budget is $B_1$. We always normalize the sum of budgets to $1$, thus, implicitly, \PT's budget is $B_2 = 1-B_1$. 
At configuration $c$, we charge and normalize the budgets. Formally, the game proceeds to an {\em intermediate configuration} $c' = \zug{v, B'_1}$ defined by $B'_1 = \frac{B_1+\RO(v)}{1+\RO(v)+\RT(v)}$. \PT's budget becomes $B'_2 = 1- B'_1 = \frac{B_2+\RT(v)}{1+\RO(v)+\RT(v)}$.
Then, the players simultaneously bid for the privilege of moving the token. Formally, for $i \in \set{1,2}$, \PLi chooses an {\em action} $\zug{b_i, u_i}$, where $b_i \in [0,B'_i]$ and $u_i \in \N(v)$. 
Given both players' actions, the next configuration is $\zug{u, B''_1}$, where $u = u_1$ when $b_1 \geq b_2$ and $u = u_2$ when $b_2 > b_1$, and $B''_1$ is determined based on the \emph{bidding mechanism} defined below. 
Note that we arbitrarily break ties in favor of \PO, but it can be shown that all our results remain valid no matter how ties are resolved. 
In the definitions below we assume that \PO is the higher bidder, i.e., $b_1 \geq b_2$, and the case where $b_2 > b_1$ is dual: 

\begin{description}
	\item[Richman bidding] The higher bidder pays his bid to the lower bidder. 
		Formally, 
			$\BO'' = \BO'-b_1$, and
			$\BT'' = \BT' + b_1$.
	\item[Poorman bidding] 
	The higher bidder pays his bid to the bank and we re-normalize the budget to sum up to $1$.
		Formally, 
		$\BO'' = \frac{\BO'-b_1}{1-b_1}$, and
		$\BT'' = \frac{\BT'}{1-b_1}$. 
	\item[Taxman bidding] For a predetermined and fixed fraction $\tau\in [0,1]$, called the \emph{tax rate}, the higher bidder pays fraction $\tau$ of his bid to the bank, and the rest to the lower bidder.
		Formally,
		$\BO'' = \frac{\BO'-b_1}{1-\tau\cdot b_1}$, and
		$\BT'' = \frac{\BT' + (1-\tau)\cdot b_1}{1-\tau\cdot b_1}$. Note that taxman bidding with $\tau=0$ coincides with Richman bidding and with $\tau=1$ coincides with poorman bidding. 
\end{description}

In a bidding game, a {\em history}  is a finite sequence $h = \zug{v_0, B_0}, \zug{v_0, B'_0},\ldots, \zug{v_n, B_n}, \zug{v_n, B'_n}$ which alternates between configurations and intermediate configurations. For $i \in \set{1,2}$, a \emph{strategy} for \PLi is a function $\stratI$ that maps a history to an action $\zug{b_i,u_i}$. 
We typically consider {\em memoryless} strategies, which are functions from intermediate configurations to actions. 
An initial configuration $c_0 = \zug{v_0, B_0}$ and two strategies $\strat_1$ and $\strat_2$ give rise to an infinite {\em play}, denoted $\play(c_0, \strat_1, \strat_2)$, and is defined inductively, where the inductive step is based on the definitions above. Let $\play(c_0, \strat_1, \strat_2) = \zug{v_0, B_0}, \zug{v_0, B'_0},\ldots$. The {\em path} that corresponds to $\play(c_0, \strat_1, \strat_2)$ is $v_0, v_1, \ldots \in V^\omega$. 

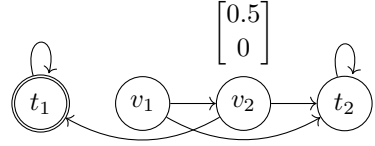
\begin{wrapfigure}{r}{0.4\textwidth}
	\centering
	\tikzset{every state/.style={minimum size=0pt}}
		\centering
		\begin{tikzpicture}[node distance=0.6cm]
			\node[state]	(v1)	at	(0,0)		{$v_1$};
			\node[state,label={above:$\begin{bmatrix}0.5\\ 0\end{bmatrix}$}]	(v2)	[right=of v1]	{$v_2$};
			\node[state,accepting]	(t1)	[left=of v1]	{$t_1$};
			\node[state]	(t2)	[right=of v2]	{$t_2$};
			
			\path[->]	(v1)	edge	(v2)
			(v1)	edge[bend right]			(t2)
			(v2)	edge[bend left] 			(t1)
			(v2)	edge 			(t2)
			(t1) edge[loop above] ()
			(t2) edge[loop above] ();
			
		\end{tikzpicture}
		\caption{Example of a poorman bidding game where without normalization thresholds are not uniquely determined.}
		\label{fig:examples:normalization}
\end{wrapfigure}
Each game is equipped with an {\em objective} $\spec \subseteq V^\omega$. 
Each play has a {\em winner}. \PO wins a play if its corresponding path is in $\spec$, and \PT wins otherwise. 
For an objective $\spec$, a \PO strategy $\strat_1$ is {\em winning} from a configuration $c$ if for every \PT strategy $\strat_2$, the play $\play(c, \strat_1, \strat_2)$ is winning for \PO, and the definition for \PT is dual. For $i \in \set{1,2}$, we say \PLi \ {\em wins} from configuration $c$ for $\varphi$ if he has a winning strategy from $c$. We will use $\compl{x}$ to denote the complement of $x$, where $x$ can be either an objective or a set of vertices.
We consider the following objectives:

\begin{description}
	\item[Reachability.] For a set of vertices $T\subseteq \V$, the reachability objective is defined as $\reach(T)\coloneqq\set{v_0v_1\ldots \in\V^\omega\mid \exists i\in\mathbb{N}\;.\;v_i\in T}$.
		Intuitively, $T$ represents the set of target vertices, and $\reach(T)$ is satisfied if $T$ is eventually visited by the given path.
	\item[Safety.] For a set of vertices $S\subseteq \V$, the safety objective is defined as $\safe(S)\coloneqq\set{v_0 v_1\ldots \in\V^\omega\mid \forall i\in\mathbb{N}\;.\;v_i\in S}$. 
		Intuitively, $S$ represents the set of safe vertices, and $\safe(S)$ is satisfied if $S$ is not left ever during the given path.
		Safety and reachability are dual to each other, i.e., $\safe(S) = \compl{\reach(\compl{ S})}$.
	\item[B\"uchi.] For a set of vertices $B\subseteq \V$, the B\"uchi objective is defined as $\buchi(B)\coloneqq \set{v_0 v_1\ldots\in\V^\omega\mid \forall i\in\mathbb{N}\;.\;\exists j>i\;.\; v_j\in B}$. 
		Intuitively, $\buchi(B)$ is satisfied if $B$ is visited infinitely often during the given path.
	\item[Co-B\"uchi.] For a set of vertices $C\subseteq \V$, the co-B\"uchi objective is defined as $\cobuchi(C)\coloneqq \set{v_0 v_1\ldots\in \V^\omega\mid \exists i\in \mathbb{N}\;.\;\forall j>i\;.\; v_j\in C}$. 
		Intuitively, $\cobuchi(C)$ is satisfied if only $C$ is visited from some point onward during the play.
		B\"uchi and co-B\"uchi objectives are dual to each other, i.e., $\cobuchi(C) = \compl{\buchi(\compl{C})}$.
\end{description}

A central concept in bidding games is the pair of {\em thresholds} for the two players.
Roughly, they are the smallest budgets needed by the respective player for winning the game from a given vertex.
We formalize this below.

\begin{definition}[Thresholds]
	Let $\G$ be a given arena and $M\in \lbrace\mathit{Richman}, \mathit{poorman}, \mathit{taxman}\rbrace$ be a given bidding mechanism .
	For an objective $\varphi$, the thresholds $\Th^{\G,M,\varphi}_1,\Th^{\G,M,\varphi}_2 \colon \V\to [0,1]$ are functions such that for every $v\in \V$ and every $\epsilon > 0$:
	\begin{itemize}
	\item $\Th^{\G,M,\varphi}_1(v) \coloneqq \inf_{B_1 \in [0,1]} \{ B_1:$ \PO wins from $\zug{v, B_1+\epsilon}$ for $\varphi$, for every $\epsilon >0\}$.
	\item $\Th^{\G,M,\varphi}_2(v) \coloneqq \inf_{B_2 \in [0,1]} \{ B_2:$ \PT wins from $\zug{v, 1-B_2-\epsilon}$ for $\compl{\varphi}$, for every $\epsilon >0\}$.
	\end{itemize}
	When $\Th^{\G,M,\varphi}_1(v) + \Th^{\G,M,\varphi}_2(v) = 1$ for every vertex $v$, we say that the {\em threshold exists} in $\G$, denote it $\Th^{\G,M,\varphi}(v)$, and define $\Th^{\G,M,\varphi}(v) = \Th^{\G,M,\varphi}_1(v)$. 
\end{definition}
Whenever the game graph and the bidding mechanism are clear from the context, we simply write $\ThO^\spec$, $\ThT^\spec$, and $\Th^\spec$.

\begin{example}
\label{ex:normalization}{\bf (The importance of normalization).}
Consider the poorman bidding game that is depicted in Fig.~\ref{fig:examples:normalization}. Intuitively, \PO wins from $v_1$ iff he wins the first two consecutive biddings.
Formally, the game starts at $v_1$ and $t_i$ is \PLi's target, for $i \in \set{1,2}$. We first analyze the game with a normalization step. We argue that $\Th(v_2) =\frac{1}{4}$; indeed, since $\frac{1/4+1/2}{3/2} = \frac{1}{2}$, entering $v_2$ with a budget greater than $\frac{1}{4}$ allows \PO to secure winning. We argue that $\Th(v_1) = \frac{4}{7}$; indeed, \PO must bid above \PT's budget and win the bidding, and note that $\frac{4/7 - 3/7}{3/7} = \frac{1}{4}$. Note that thresholds are in fact a {\em ratio}. Stated differently, consider a configuration $\zug{v_1, B_1, B_2}$ with $B_1+B_2$ not necessarily equals $1$, then \PO wins iff $\frac{B_1}{B_1 + B_2} > \frac{4}{7}$. Crucially, the ratio between $B_1$ and $B_2$ is fixed. We will prove that this is a general phenomenon on which our algorithms depend. 

When a normalization step is \emph{not} performed, \PO's threshold for winning is a non-linear function of \PT's initial budget. Intuitively, when no normalization is performed, the charge is more meaningful when the budgets are smaller.
Consider a configuration $\zug{v_1, B_1, B_2}$ with $B_1+B_2$ being not necessarily equal to $1$. 
Note that \PO must win the first bidding, thus the second configuration must be $\zug{v_2, B_1 - B_2, B_2}$. When no normalization is performed after charging, the intermediate configuration is $\zug{v_2, B_1 - B_2 + 0.5, B_2}$. Clearly, \PO wins iff $B_1 - B_2 + 0.5 > B_2$. For example, when $B_2 = 1$, then \PO's threshold is $\frac{3}{2}$ and when $B_2 = 2$, then \PO's threshold is $\frac{7}{2}$. These amount to ratios of $\frac{3}{5}$ and $\frac{7}{11}$, respectively, meaning that \PO's threshold is a non-linear function of \PT's budget. We point out that this is also the case in poorman discrete-bidding games~\cite{AM+23}, where thresholds can only be approximated, even in extremely simple games.
\qed
\end{example}

We formulate the decision problem related to the computation  of thresholds.
We will write that a given objective \emph{$\spec$ is of type} $\reach$, $\safe$, $\buchi$, or $\cobuchi$ if $\spec$ can be expressed as a reachability, safety, B\"uchi, or co-B\"uchi objective (on a given arena), respectively.

\begin{definition}[Finding threshold budgets]
Let $M\in \set{\mathit{Richman},\mathit{poorman},$ $\mathit{taxman}}$, and  $S\in \set{\reach,\safe,\buchi,\cobuchi}$.
The problem $\THRESH^M_S$ takes as input an arena $\G$, an initial vertex $v$, and an objective $\varphi \in S$, and accepts the input iff $\ThO^{\G,M,\varphi}(v) \leq 0.5$.
\end{definition}

%!TEX root=main.tex

%%%%%%%%%%%%%%%%%%%%%%%%%%%%%%%%%%%%%%
\section{Reachability Bidding Games with Charging}
\label{sec:reachability}

In this section, we show the existence of thresholds in taxman-bidding games with charging with reachability and, dually, with safety objectives. 
Throughout this section, we fix an arena $\G=\tup{\V,\E,\RO,\RT}$. 
For a given set of vertices $T\subseteq \V$, the objective of \PO, the reachability player, is $\reach(T)$, and, the objective of \PT, the safety player, is $\safe(\compl{T})$. 

\subsection{Bounded-Horizon Reachability and Safety}
We start with the simpler case of {\em bounded-horizon} reachability objectives, and in the next section, we will extend the technique to general games. 
Let $t \in \Nat$. The bounded-horizon reachability, denoted $\reach(T, t)$, intuitively requires \PO to reach $T$ within $t$ steps. Formally, $\reach(T,t)\coloneqq\set{v_0 v_1\ldots\mid \exists i\leq t\;.\;v_i\in T}$. {\em Bounded-horizon safety} is the dual objective $\safe(\compl{T},t)\coloneqq\set{v_0 v_1\ldots\mid\forall i\leq t\;.\;v_i\notin T} = V^\omega\setminus \reach(T,t)$. 

In the following, we characterize the thresholds for $\reach(T, t)$ and $\safe(\compl{T}, t)$ by induction on $t$. The induction step relies on the following operator on functions.

\begin{definition}
	\label{def:taxman-upd}
	Define the function $\clamp{x} \coloneqq\min(1,\max(0,x))$; that is, given $x$, $\clamp{x} = x$, when $0 < x < 1$, and otherwise it ``saturates'' $x$ at the boundaries $0$ or $1$. 
	Let $\tau\in [0,1]$ be the tax rate.
	 We define two operators on functions $\updO,\updT \colon [0,1]^V \to [0,1]^V$ as follows. For $i \in \set{1,2}$ and $f\in [0,1]^\V$:
	\[\updI(f)(v) \coloneqq \clamp{\frac{(1-\tau)f(v^-)+f(v^+)}{[f(v^+)-f(v^-)-1]\tau + 2} \cdot (1+\RO(v)+\RT(v)) - \RI(v)}\]
	where $v^+$ and $v^-$ are the successors of $v$ with the largest and the smallest value of $f(\cdot)$, respectively, i.e., $v^+ = \arg\max_{u \in \N(v)} f(u)$ and $ v^- = \arg\min_{u \in \N(v)} f(u)$.
\end{definition}

Note that for Richman bidding, i.e., when $\tau = 0$, for $i \in \set{1,2}$, we have 
\[
\updI(f)(v) \coloneqq \clamp{\frac{f(v^+)+f(v^-)}{2} \cdot \big(1+R_1(v)+R_2(v)\big)- \RI(v)}
\]
In this case, the function $\updI$ computes the average (the name ``$\mathit{Av}$'' stands for ``average'') of its argument $f$ on $v^-$ and $v^+$, and then performs an affine transformation followed by the saturation $\clamp{\cdot}$ on the result.
For poorman bidding, i.e., when $\tau = 1$, we have
	\[
	\updI(f)(v) \coloneqq \clamp{\frac{f(v^+)}{f(v^+)-f(v^-)+1} \cdot \big(1+\RO(v)+\RT(v)\big) - \RI(v)}
	\]
\stam{
	\begin{align*}
		&\updO(f)(v) \coloneqq\\
		&\quad \clamp{\frac{f(v^+)+f(v^-)}{2}\big(1+R_1(v)+R_2(v)\big)- \RO(v)},\\
		&\updT(f)(v) \coloneqq\\
		&\quad \clamp{\frac{f(v^+)+f(v^-)}{2}\big(1+R_1(v)+R_2(v)\big)- \RT(v)},
	\end{align*}
}

We define two functions $f_1$ and $f_2$ which will be shown to coincide with the thresholds. 

\begin{definition}
	\label{def:richman-reach-th}
	Define the functions $f_1, f_2\colon \V\times \mathbb{N}\to [0,1]$ inductively on $t$. For every $v\in T$ and $t\in \Nat$, define $f_1(v,t) \coloneqq 0$ and $f_2(v, t) \coloneqq 1$. 	
	For every $ v\notin T$, define $f_1(v,0) \coloneqq 1$ and $f_2(v,0) \coloneqq 0$, and for every $t > 0$, define $f_1(v,t) \coloneqq \updO\left( f_1(\cdot,t-1) \right)(v)$ and $f_2(v,t) \coloneqq \updT\left(f_2(\cdot,t-1)\right)(v)$.
\end{definition}

Lem.~\ref{lem:soundness of time-varying threshold} shows that $f_1$ and $f_2$ coincide with the thresholds of the (bounded-horizon) reachability and safety players, respectively. 
Intuitively, for $\reach(T,0)$, \PO wins with even zero budget from vertices that are already in $T$, and loses with even the maximum budget from vertices that are not in $T$.
We capture this as $f_1(v,0)=0$ if $v\in T$, and $f_1(v,0)=1$ otherwise. Furthermore, if \PO has a budget more than $f_1(v,t)$ at $v$, then we show that he has a memoryless policy such that no matter which vertex $v'$ the token reaches in the next step, his budget will remain more than $f_1(v',t-1)$. It follows inductively that he will reach $T$ in $t$ steps from $v$. The argument for the safety player is dual. Lem.~\ref{lem:soundness of time-varying threshold} also establishes the existence of thresholds.

\begin{restatable}[]{lemma}{lemExistanceThreshBoundedReach}\label{lem:soundness of time-varying threshold}
For every vertex $v \in V$ and $t \geq 0$, we have $\Th_1^{\reach(T, t)}(v)=f_1(v,t)$ and $\Th_2^{\reach(T, t)}(v)=f_2(v,t)$. Moreover, thresholds exist:  $\Th_1^{\reach(T, t)}(v) + \Th_2^{\reach(T, t)}(v) = 1$.
\end{restatable}

\begin{proof}
We provide separate proofs for the Richman and taxman bidding mechanisms: the former is for easier intuitive explanation and the latter is for completeness.

\smallskip
\noindent\textbf{Proof for the simpler case of Richman bidding.}
	A straightforward arithmetic verification shows that $f_1(v, t) =1- f_2(v, t)$. 
	To show $f_1(v,t) \geq \Th_1^{\reach(T, t)}(v)$, we describe a winning \PO strategy from configuration $\zug{v, f_1(v, t) + \epsilon}$ for an arbitrary $\epsilon > 0$. 
	It is dual to show $f_2(v, t) \geq \Th_2^{\reach(T, t)}(v)$. 
	Equality then follows from the relationship $f_1(v, t) = 1 - f_2(v, t)$.
	
	The proof of $f_1(v,t) \geq \Th_1^{\reach(T, t)}(v)$ proceeds by induction over $t$. The base case is $t=0$. Clearly, for $v \in T$, it holds that $\Th_1^{\reach(T, 0)}(v) = f_1(v,0) = 0$ and for $v \notin T$, that $\Th_1^{\reach(T, 0)}(v) = f_1(v,0) = 1$. 
	
	For $t \geq 1$, assume that $f_1(v,t-1) \geq \Th_1^{\reach(T, t-1)}(v)$, and we prove the claim for $t$. 
	Let $v \in V$ and $B_1>f_1(v, t)$. We describe a \PO winning strategy from $\zug{v, B_1}$.
	\PO bids $b_1 = \frac{f_1(v^+, t-1) - f_1(v^-, t-1)}{2}$, and recall that $v^+$ and $v^-$ are the successors of $v$ that attain the maximal and minimal value of $f_1(\cdot, t-1)$. Upon winning, \PO proceeds to $v^-$. 
	Let $\zug{v', B'_1}$ be the next configuration given some \PT action. We claim that \PO's strategy guarantees that $B'_1 > f_1(v', t-1)$. Thus, by the induction hypothesis, \PO reaches $T$ in at most $t$ steps and the claim follows. 
	
	To see this, first observe that, by definition, $f_1(v, t)$ can be either $0$, $1$, or $x \coloneqq \frac{f_1(v^+,t-1)+f_1(v^-,t-1)}{2}\big(1+R_1(v)+R_2(v)\big)- R_1(v)$.
	However, since $\BO\in [0,1]$ and $\BO>f_1(v,t)$, we can rule out the case of $f_1(v,t)=1$.
	Therefore, we have $f_1(v,t) = \max(0,x)\geq x$. We distinguish between two cases. 
	First, suppose that \PO wins the bidding, thus he pays his bid to the opponent and moves the token to $v^-$.
	His updated budget at $v^-$ becomes:
	\begin{align*}
		\frac{\BO+\RO(v)}{1+\RO(v)+\RT(v)} - \bo 
		&> \frac{f_1(v,t) + \RO(v)}{1+\RO(v)+\RT(v)} - \bo\\
		&\geq \frac{x + \RO(v)}{1+\RO(v)+\RT(v)} - \bo\\
		&= \frac{f_1(v^+,t-1)+f_1(v^-,t-1)}{2}\\ 
		&\quad- \frac{f_1(v^+,t-1)-f_1(v^-,t-1)}{2} \\ &= f_1(v^-,t-1).
	\end{align*}
	Second, suppose that \PO loses the bidding, thus \PT pays him at least $b_1$ and the token moves to $v' \in \N(v)$. Recall that $f_1(v')\leq f_1(v^+)$.
	The updated budget of \PO at $v'$ becomes:
	\begin{align*}
		\frac{\BO+\RO(v)}{1+\RO(v)+\RT(v)} + \bo 
		&> \frac{f_1(v,t) + \RO(v)}{1+\RO(v)+\RT(v)} + \bo\\
		&\geq \frac{x + \RO(v)}{1+\RO(v)+\RT(v)} + \bo\\
		&= \frac{f_1(v^+,t-1)+f_1(v^-,t-1)}{2}\\
		&\quad + \frac{f_1(v^+,t-1)-f_1(v^-,t-1)}{2}\\
		&= f_1(v^+,t-1)\geq f_1(v',t-1).
	\end{align*}

\smallskip
\noindent\textbf{Proof for taxman bidding.}
	Consider a taxman bidding game $\G$ with tax rate $\tau$ and reachability objective $\reach(T)$.
	
	We first show that $f_1(v,t)+f_2(v,t)=1$. We use induction on $t$: For $t=0$ the following holds:
	\[
	\begin{split}
		\forall v \in T,~ & f_1(v,0)+f_2(v,0) = 1 \\
		\forall v \notin T,~ & f_1(v,0)+f_2(v,0) = 1 \\
	\end{split}
	\]
	For $t>0$, by induction hypothesis, suppose that $f_1(v,t-1)+f_2(v,t-1)=1$ for all $v \in V$. Then, for $v \in T$, we trivially have $f_1(v,t)+f_2(v,t)=1$. For $v \notin T$, the following holds:
	\[
	\begin{split}
		&\frac{(1-\tau)f_2(v^-,t-1) + f_2(v^+,t-1)}{[f_2(v^+,t-1)-f_2(v^-,t-1)-1]\tau +2}(1+\RO(v)+\RT(v)) - \RT(v) \\
		+&  \frac{(1-\tau)f_1(v^-,t-1) + f_1(v^+,t-1)}{[f_1(v^+,t-1)-f_1(v^-,t-1)-1]\tau +2}(1+\RO(v)+\RT(v))-\RO(v)\\
		=& \frac{(1-\tau)f_2(v^-,t-1) + f_2(v^+,t-1)}{[f_2(v^+,t-1)-f_2(v^-,t-1)-1]\tau +2}(1+\RO(v)+\RT(v)) - \RT(v) \\
		+ & \frac{-(1-\tau)f_2(v^+,t-1) - f_2(v^-,t-1) + 2-\tau}{[f_2(v^+,t-1)-f_2(v^-,t-1)-1]\tau +2}(1+\RO(v)+\RT(v))-\RO(v)\\
		=& (1+\RO(v)+\RT(v))-\RT(v) - \RO(v) = 1 \\
	\end{split}
	\]
	It follows easily that $f_1(v,t)+f_2(v,t)=1$. 
	
	Next, we show that $f_1(v,t) \leq \ThO^{\reach(T,t)}(v)$ for all $v$, by describing a winning strategy of \PO from configuration $\tup{v,f_1(v,t)+\epsilon}$. We again use induction on $t$. For $t=0$ the statement holds trivially, i.e. $f_1(v,0)=0$ for $v \in T$ and $f_1(v,0)=1$ for $v \notin T$. For $t>0$, suppose \PO has budget $B= f_1(v,t)+\epsilon$. He can bid $b=\frac{f_1(v^+,t-1)-f_1(v^-,t-1)}{[f_1(v^+,t-1)-f_1(v^-,t-1)-1]\tau +2}$. In case of winning he will move the token to $v^-$ and have $B'$ budget left where
	\[
	\begin{split}
		B' =~& \frac{\frac{B+\RO(v)}{1+\RO(v)+\RT(v)}-b}{1-\tau b} \\
		>~& \frac{\frac{(1-\tau)f_1(v^-,t-1) + f_1(v^+,t-1)}{[f_1(v^+,t-1)-f_1(v^-,t-1)-1]\tau +2} - \frac{f_1(v^+,t-1)-f_1(v^-,t-1)}{[f_1(v^+,t-1)-f_1(v^-,t-1)-1]\tau +2}}{\frac{2-\tau}{[f_1(v^+,t-1)-f_1(v^-,t-1)-1]\tau +2}} \\
		=~ &\frac{(1-\tau)f_1(v^-,t-1) + f_1(v^+,t-1) - f_1(v^+,t-1)+f_1(v^-,t-1)}{2-\tau}\\
		=~& f_1(v^-,t-1)
	\end{split}
	\]
	So, by induction hypothesis he makes the game reach $T$ from $v^-$ in at most $t-1$ steps. 
	
	In case of losing, the token will be moved to some successor $u$ of $v$ where $f_1(u,t-1) \leq f_1(v^+,t-1)$ and he will have $B'$ budget left where:
	\[
	\begin{split}
		B' >~& \frac{\frac{B+\RO(v)}{1+\RO(v)+\RT(v)}+(1-\tau)b}{1-\tau b} \\
		>~& \frac{\frac{(1-\tau)f_1(v^-,t-1) + f_1(v^+,t-1)}{[f_1(v^+,t-1)-f_1(v^-,t-1)-1]\tau +2} + \frac{(1-\tau)[f_1(v^+,t-1)-f_1(v^-,t-1)]}{[f_1(v^+,t-1)-f_1(v^-,t-1)-1]\tau +2}}{\frac{2-\tau}{[f_1(v^+,t-1)-f_1(v^-,t-1)-1]\tau +2}}\\
		=~ & \frac{(1-\tau)f_1(v^-,t-1) + f_1(v^+,t-1) + (1-\tau)[f_1(v^+,t-1)-f_1(v^-,t-1)]}{2-\tau}\\
		=~& f_1(v^+,t-1) \\
		\geq~ & f_1(u,t-1)
	\end{split}
	\]
	Again, by induction hypothesis, he will have enough budget to reach $T$ in at most $t-1$ steps from $u$.
	
	For $f_2(v,t)$, \PT has a similar strategy based on $f_2(v,t-1)$. Note that at each point of the game, she is maintaining $B_2 > f_2(v,t)$, implying that $f_2(v,t) < 1$ which means $v \notin T$. So, if she has more budget than $f_2(v,t)$, she can prevent $T$ from being reached for at least $t$ steps.
\end{proof}

The following lemma establishes monotonicity of $f_1$ and $f_2$ with respect to $t$, which will play a key role in the proof of existence of thresholds for the unbounded counterparts of the objectives. 
Intuitively, reaching $T$ within $t$ steps is harder than reaching $T$ within $t' > t$ turns, thus less budget is needed for the latter case. Dually, guaranteeing safety for $t$ turns is easier than guaranteeing safety for $t' > t$ turns.

\begin{restatable}{lemma}{lemMonotonicThreshBoundedReach}
\label{lem:monotonic}
For $v\in\V$ and $t'> t$, it holds that $f_1(v,t')\leq f_1(v,t)$ and $f_2(v,t')\geq f_2(v,t)$.
\end{restatable}
\begin{proof}%[Proof of Lem. \ref{lem:monotonic}]
	Let $0 \leq \tau \leq 1$ and Define function $P_\tau: [0,1] \times [0,1] \to \Real$ as follows:
	\[
	P_\tau(x,y) = \frac{(1-\tau)x + y}{(y-x-1)\tau+2}
	\]
	The following holds:
	\[
	\begin{split}
		\frac{\partial P_\tau}{\partial x} &= \frac{(1-\tau)[(y-x-1)\tau +2] + \tau [(1-\tau) x + y]}{[(y-x-1)\tau+2]^2} \\
		&= \frac{\tau(2-\tau)y+(1-\tau)(2-\tau)}{[(y-x-1)\tau+2]^2} \geq 0\\
		\frac{\partial P_\tau}{\partial y} & = \frac{[(y-x-1)\tau+2] - \tau [(1-\tau)x+y]}{[(y-x-1)\tau+2]^2}\\
		& = \frac{-\tau(2-\tau)x + (2-\tau)}{[(y-x-1)\tau+2]^2} = \frac{(2-\tau)(1-\tau x )}{[(y-x-1)\tau+2]^2} \geq 0
	\end{split}
	\]
	We prove monotonicity of $f_1(v,t)$ by induction on $t$. For $t=1$, the statement trivially holds. For $t>1$, It follows by definition that
	\[
	\begin{split}
		f_1(v,t+1) =& \clamp{P_\tau\big(f_1(v^-,t), f_1(v^+,t)\big)\times
			\big(1+\RO(v) + \RT(v)\big) - \RO(v)}\\
		f_1(v,t) =& \clamp{P_\tau\big(f_1(v^-,t-1), f_1(v^+,t-1)\big)\times
			\big(1+\RO(v) + \RT(v)\big) - \RO(v)}\\
	\end{split}
	\]
	From induction hypothesis it follows that $f_1(v^+,t) \leq f_1(v^+,t-1)$ and $f_1(v^-,t) \leq f_1(v^-,t-1)$. Therefore, $P_\tau\big(f_1(v^-,t), f_1(v^+,t)\big) \leq P_\tau\big(f_1(v^-,t-1), f_1(v^+,t-1)\big)$ which directly implies $f_1(v,t+1) \leq f_1(v,t)$.
	
	As shown in Lem. \ref{lem:soundness of time-varying threshold}, it holds that $f_1(v,t)+f_2(v,t)=1$, therefore, $f_2(v,t)$ is increasing in $t$.
\end{proof}

\subsection{Existence of Thresholds (for Reachability and Safety Objectives)}

We define two functions $f_1^*$ and $f_2^*$ which will be shown to coincide with the thresholds for the unbounded horizon reachability and safety objectives, respectively.

\begin{definition}
	\label{def:richman-reach-th-limit}
	Define the functions $f_1^*,f_2^*\colon \V\to [0,1]$, such that for every $v\in \V$:
\[
		f_1^*(v) \coloneqq \lim_{t\to\infty} f_1(v,t) \quad \text{ and } \quad
		f_2^*(v) \coloneqq \lim_{t\to\infty} f_2(v,t).
\]
\end{definition}

Since $f_1$ and $f_2$ are bounded in $[0,1]$ and monotonic by Lem.~\ref{lem:monotonic}, the limits in Def.~\ref{def:richman-reach-th-limit} are well defined.
Since $f_1(v,0)$ and $f_2(v,0)$ assign, respectively, the maximum (i.e., $1$) and the minimum (i.e., $0$) value to every vertex $v\notin T$, hence from the Kleene fixed point theorem, it follows that $f_1^*$ and $f_2^*$ will be, respectively, the greatest and the least fixed points of the operators $\updO$ and $\updT$ on the directed-complete partial order $\tup{[0,1]^\V,\leq}$.

\begin{proposition}\label{prop:richman:reach:fixpoint}
	Consider the directed-complete partial order $L=\tup{[0,1]^{\V},\leq}$, where for every $x,y\in [0,1]^{\V}$, $x\leq y$ iff $x_i\leq y_i$ for every $i\in \V$.
	The functions $f_1^*$ and $f_2^*$ are, respectively, the greatest and the least fixed points of the functions $\updO$ and $\updT$ on $L$, subjected to the constraints $f_1^*(v)=0$ and $f_2^*(v)=1$ for every $v\in T$.
\end{proposition}

The following example demonstrates that, unlike traditional bidding games, the fixed points of the functions $\updO$ and $\updT$ on $L$ may not be unique.

\begin{figure}
\begin{minipage}[c]{0.45\linewidth}
%\begin{figure}%{r}{0.7\textwidth}
%	\centering
	\begin{tikzpicture}[node distance=0.6cm]
		\node[state]	(a)	at	(0,0)		{$a$};
		\node[state]	(b)	[right=of a]	{$b$};
		\node[state,accepting]	(c)	[left=of a]	{$c$};
		\node[state,label={right:$\begin{bmatrix}0\\ 5\end{bmatrix}$}]	(d)	[right=of b]	{$d$};
		
		\path[->]	(a)	edge	(b)
						edge[bend left]	(c)
					(b)	edge[bend right]			(c)
						edge[bend left]	(d)
					(c)	edge[loop left]()
					(d)	edge[bend left]	(b)
						edge[bend right,out=45,in=135]	(c);
					
	\end{tikzpicture}
	\caption{Non-unique fixed points of the threshold update functions.}
	\label{fig:example:non-unique fixed point}
%\end{figure}
\end{minipage}\hfill
\begin{minipage}[c]{0.45\linewidth}
%\begin{table}
\begin{tabular}{c| c c c c c}
	\backslashbox{$t$}{$v$} & $a$ & $b$ & $c$ & $d$ & $e$ \\ 
	\hline
	0 & 1 & 1 & 1 & 0 & 1\\
	1 & 1 & 1  & 0.5  & 0 & 1\\
	2 & 1 & 0.75  & 0.5  & 0 & 1\\
	3 & 0.625  & 0.75  & 0.5  & 0 & 1\\
	4 & 0.0625  & 0.5625  &  0.5 & 0 & 1\\
	5 & 0  & 0.28125  & 0.5  & 0 & 1\\
	\hline 
	$\geq 6$ & 0 & 0.25 & 0.5 & 0 & 1
\end{tabular}
	\caption{Finite-Horizon reachability thresholds for the game in Fig. \ref{fig:examples:a}. The number at row $i$ and column $v$ represents $\ThO^{\reach(d,i)}(v)$.}
	\label{table:example:thresholds}
%\end{table}
\end{minipage}
\end{figure}

\begin{example}[Multiple fixed-points]
	Consider the bidding game in Fig.~\ref{fig:example:non-unique fixed point}, where the objective of \PO is to reach $c$.
	It can be easily verified that both ${f_1^*}=\set{a\mapsto 0.25, b\mapsto 0.5, c\mapsto 0, d\mapsto 1}$ and ${f_1^*}'\equiv 0$ are fixed points of the operator $\updO$ over $L$ in this case.	
	\qed
\end{example}

The following theorem establishes the existence of thresholds. 

\begin{restatable}{theorem}{thmRichmanReachComplThresh}\label{thm:richman:reach:complementary threshold}
	For every vertex $v \in V$, it holds that $\Th_1^{\reach(T)}(v)=f_1^*(v)$ and $\Th_2^{\reach(T)}(v) = f_2^*(v)$. 
Moreover, thresholds exist: $\Th_1^{\reach(T)}(v) + \Th_2^{\reach(T)}(v) = 1$.
\end{restatable}
\begin{proof}
We prove that $f_1^*(v)\geq \Th_1^{\reach(T)}(v)$, for every $v \in \V$. Let $B_1 > f_1^*(v)$. There exists a $t \in \Nat$ such that $B_1 > f_1(v, t)$. 
\PO uses the strategy from Lem.~\ref{lem:soundness of time-varying threshold}, guaranteeing that $T$ is reached within $t$ steps. 

Next, we prove that  $f_2^*(v)\geq \Th_2^{\reach(T)}(v)$, for every $v \in \V$. We provide separate proofs for the Richman and taxman bidding mechanisms: the former is for easier intuitive explanation and the latter is for completeness.

\smallskip
\noindent\textbf{Proof for Richman Bidding. }
Let $B_2 > f_2^*(v)$ and we describe a  \PT strategy that wins for the objective $\safe(\compl{T})$. 
We know that $f_2^* = \updT(f_2^*)$ (consequence of Prop.~\ref{prop:richman:reach:fixpoint}), and let $v^+,v^-$ be the successors of $v$ with the greatest and the least value of $f_2$.
\PT bids $\bt = \frac{f_2^*(v^+)-f_2^*(v^-)}{2}$, and upon winning proceeds to $v^-$.
We prove the invariant that regardless of whether \PT wins or loses the bidding, in the next vertex $v'$ her new budget $B_2'>f_2^*(v')$.
Since, by construction, $f_2^*(v'')=1$ for every $v''\in T$, hence it must be true that $T$ is never reached as \PT's budget can never exceed $1$.

Before we prove the invariant condition, observe that $f_2^*(v)$ can be either $0$, $1$, or $x = \frac{f_2^*(v^+)+f_2^*(v^-)}{2}\left( 1+R_1(v)+R_2(v)\right) -R_1(v)$.
Since $\BT\in [0,1]$ and $\BT>f_2^*(v)$, we can rule out the case $f_2^*(v)=1$.
Therefore $f_2^*(v)\geq x$.

To show the invariant, first, assume that \PT wins the bidding and proceeds to $v^-$.
Her new budget at $v^-$ becomes:
\begin{align*}
	\BT' = \frac{\BT+\RT(v)}{1+\RO(v)+\RT(v)} - \bt 
	&> \frac{f_2^*(v) + \RT(v)}{1+\RO(v)+\RT(v)} - \bt\\
	&\geq \frac{x + \RT(v)}{1+\RO(v)+\RT(v)} - \bt\\
	&= \frac{f_2^*(v^+)+f_2^*(v^-)}{2} \\
	&\quad -\frac{f_2^*(v^+)-f_2^*(v^-)}{2} \\ &= f_2^*(v^-).
\end{align*}
Next, assume that \PT loses the bidding and \PO moves the token to some $v'\in\N(v)$.
By definition, $f_2^*(v')\leq f_2^*(v^+)$.
The new budget of \PT at $v'$ becomes:
\begin{align*}
	\BT' = \frac{\BT+\RT(v)}{1+\RO(v)+\RT(v)} + & \bt > \frac{f_2^*(v) + \RT(v)}{1+\RO(v)+\RT(v)} + \bt\\
	&\geq \frac{x + \RT(v)}{1+\RO(v)+\RT(v)} + \bt\\
	&= \frac{f_2^*(v^+)+f_2^*(v^-)}{2} + \frac{f_2^*(v^+)-f_2^*(v^-)}{2}\\
	&= f_2^*(v^+)\geq f_2^*(v').
\end{align*}%

\smallskip
\noindent\textbf{Proof for taxman Bidding.} Suppose the token is in node $v$ and \PT has budget $B_2 > f_2^*(v)$. It immediately follows that $f_2^*(v) <1$. She can make a bid equal to $b=\frac{f_2(v^+)-f_2(v^-)}{[f_2(v^+)-f_2(v^-)-1]\tau +2}$. Similar to statements in the proof of Lem. \ref{lem:soundness of time-varying threshold}, it can be shown that in case of winning the bid, her budget will be strictly greater than $f_2^*(v^-)$ and in case of losing, strictly greater than $f_2^*(v^+)$. Therefore, by playing action $\tup{v^-,b}$, she is guaranteed to have move than $f_2^*(u)$ for any successor $u$ of $v$ that is reached next. It then follows that for any such $u$, $f_2^*(u)<1$, therefore $u \notin T$. Thus, $T$ is never reached since \PT's budget can never exceed $1$. Finally, the other directions, i.e., the inequalities $f_1^*(v)\leq \Th_1^{\reach(T)}(v)$ and $f_2^*(v)\leq \Th_2^{\reach(T)}(v)$, and the existence claim follows from the observation that $f_1^*(v) = 1-f_2^*(v)$, for every $v \in \V$.
\end{proof}

It follows that the threshold can be computed using fixed point iterations sketched in Prop.~\ref{prop:richman:reach:fixpoint}; this iterative approach is illustrated in the following example.

\begin{example}
	Consider the bidding game in Fig. \ref{fig:examples:a} with Richman bidding. The finite horizon reachability thresholds are depicted in table \ref{table:example:thresholds}. Suppose the game starts from $\tup{a,0.1}$ which is winning for the reachability player. 
	In this case, $t=4$ is the smallest integer for which \PO's budget $0.1$ at $a$ is larger than $\ThO^{\reach(d,t)}(a)$, which is $0.0625$.	
	Therefore, according to the strategy of \PO as described in the proof of Lem.~\ref{lem:soundness of time-varying threshold}, \PO has a strategy for reaching $d$ in $4$ steps.
	 First, his budget is charged to $\frac{0.1+2}{3}=0.7$. His winning strategy (described in the proof of Lem.~\ref{lem:soundness of time-varying threshold}) dictates that he should bid $\frac{\ThO^{\reach(d,3)}(b)-\ThO^{\reach(d,3)}(a)}{2} = 0.0625$. In case of winning, he pays the bid to the safety player and keeps the token at $a$, with budget $0.6375$ which is then charged to $0.8791\bar{6}$. This is enough for winning 3 consecutive biddings and moving the token to $d$. In case of losing the first bid, his budget will increase to at least $0.76$. This amount is more than both $\ThO^{\reach(d,3)}(a)$ and $\ThO^{\reach(d,3)}(b)$, therefore he can guarantee a win in at most 3 steps. 
	
	Now suppose the game starts from $\tup{b,0.2}$ meaning that the safety player has $0.8$ budget and can win the game. She has a budget-agnostic strategy which dictates her to bid $0.25$ in $b$ (see the proof of Lem.~\ref{lem:soundness of time-varying threshold} for a sketch of \PT's strategy). She definitely wins this bidding as her opponent has only $0.2$ budget. She then moves the token to $c$ and pays $0.25$ to the reachability player, leaving her with $0.55$ of the total budget. She can then bid $0.5$, win the bidding and move the token to $e$, where the token stays indefinitely. \qed
\end{example}

\subsection{Complexity Bounds (for Reachability and Safety Objectives)}

Since $f_1^*$ and $f_2^*$ are fixed points of the operators $\updO$ and $\updT$, respectively, hence $f_1^* = \updO(f_1^*)$ and $f_2^*=\updT(f_2^*)$.
Moreover, $f_1^*$ is the greatest fixed point, which means that $\ThO^{\reach(T)}=f_1^*$ can be computed by finding the element-wise maximum function $h$ in $[0,1]^\V$ that satisfies $h(v)=0$ for $v\in T$ and $h(v) =\updO(h)(v)$ for $v\notin T$.
This is formalized below:
\begin{align}\label{eq:MILP for richman reachability}
	\begin{split}
    &\qquad\max_h \sum_{v \in V} h(v) \\
    &\text{subjected to constraints:}\\
     &\qquad\forall v\in T\;.\; h(v)= 0, \\
     &\qquad\forall v\notin T\;.\; h(v) = \updO\left(h(\cdot)\right)(v) \\
     &\qquad\qquad= \clamp{\frac{(1-\tau)h(v^-)+h(v^+)}{[h(v^+)-h(v^-)-1]\tau + 2}\cdot
(1+\RO(v)+\RT(v)) - \RO(v)},\\
     &\qquad h(v^+) = \max_{u \in \N(v)} h(u),\quad h(v^-) = \max_{u \in \N(v)} h(u).
     \end{split}
\end{align}

\begin{restatable}{proposition}{propRichmanReachSoundnessMILP}\label{prop:richman:reach:sound milp}
	The solution of the optimization problem in \eqref{eq:MILP for richman reachability} is equivalent to the threshold function $\ThO^{\reach(T)}$ of the reachability player.
\end{restatable}

\begin{proof}%[Proof of Prop. \ref{prop:richman:reach:sound milp}]
	Prop. \ref{prop:richman:reach:fixpoint} shows that $f_1^*$ is the greatest fixed point of $\updO$. The optimization problem in \eqref{eq:MILP for richman reachability} searches for a fixed point of $\updO$ with maximum sum of values. This trivially coincides with the greatest fixed point of $\updO$.
\end{proof}

Due to Thm.~\ref{thm:richman:reach:complementary threshold}, for every vertex $v\in \V$, we have $\ThT^{\reach(T)}(v)=1-\ThO^{\reach(T)}(v)$.
Consequently, we obtain the following upper complexity bounds.

\begin{restatable}{theorem}{ThmReachTaxmanPspace} \label{thm:reach:taxman:pspace}
	The following hold:
	\begin{enumerate}[(i)]
		\item $\THRESH^{\mathrm{taxman}}_{\reach}\in \textup{PSPACE}$ and $\THRESH^{\mathrm{taxman}}_{\safe}\in \textup{PSPACE}$, \label{complexity:taxman:reach}
		\item $\THRESH^{\mathrm{Richman}}_{\reach}\in \textup{coNP}$ and $\THRESH^{\mathrm{Richman}}_{\safe}\in \textup{NP}$. \label{complexity:richman:reach}
	\end{enumerate}
\end{restatable}
\begin{proof}%[Proof of Thm. \ref{thm:reach:taxman:pspace}]
	\textbf{Proof of (i).}
	Consider the following system of inequalities: 
		\[
		\begin{split}
			&\big(\forall u \in T,~  h(u)=0 \big) \wedge A^\G_T \\
			%&\forall u \notin T: \\
			%&  \big[h(u)=0 \wedge \frac{(1-\tau)h(u^-)+h(u^+)}{[h(u^+) - h(u^-) -1] \tau +2}(1+\RO(u)+\RT(u))-\RO(u) \leq 0 \big] \\
			%&\vee  \big[h(u)=1 \wedge \frac{(1-\tau)h(u^-)+h(u^+)}{[h(u^+) - h(u^-) -1] \tau +2}(1+\RO(u)+\RT(u))-\RO(u) \geq 1 \big] \\
			%&\vee  \big [0 \leq h(u) \leq 1 \\ 
			%& ~~~~~~~~~~~~~~~~~\wedge  h(u) = \frac{(1-\tau)h(u^-)+h(u^+)}{[h(u^+) - h(u^-) -1] \tau +2}(1+\RO(u)+\RT(u))-\RO(u) \big]\\
			%& \forall u \in V, \Big[ \bigvee_{u' \in \N(u)} h(u^+)=h(u') \wedge \bigvee_{u' \in \N(u)} h(u^-)=h(u') \\
			%& ~~~~~~~~~~\wedge \forall u'\in \N(u),~ h(u^+)\geq h(u') \wedge h(u^-) \leq h(u') \Big]
		\end{split}
		\]
	With $A^\G_T$ defined as in Equation (\ref{eq:definition:AGT}). 
	\begin{equation}	\label{eq:definition:AGT}
		\begin{split}
			A^\G_T \equiv& \Bigg( \forall u \in \V\setminus T\;.\\
			&\left[h(u)=0 \wedge \frac{(1-\tau)h(u^-)+h(u^+)}{[h(u^+) - h(u^-) -1] \tau +2}(1+\RO(u)+\RT(u))-\RO(u) \leq 0 \right] \\
			&\vee  \left[h(u)=1 \wedge \frac{(1-\tau)h(u^-)+h(u^+)}{[h(u^+) - h(u^-) -1] \tau +2}(1+\RO(u)+\RT(u))-\RO(u) \geq 1 \right] \\
			&\vee  \left[0 \leq h(u) \leq 1 \wedge  h(u) = \frac{(1-\tau)h(u^-)+h(u^+)}{[h(u^+) - h(u^-) -1] \tau +2}(1+\RO(u)+\RT(u))-\RO(u) \right] \Bigg)\\
			& \wedge \Bigg( \forall u \in V\;.\; \Big[ \bigvee_{u' \in \N(u)} h(u^+)=h(u') \wedge \bigvee_{u' \in \N(u)} h(u^-)=h(u') \wedge\\
			&\qquad\qquad\qquad \forall u'\in \N(u)\;.\; h(u^+)\geq h(u') \wedge h(u^-) \leq h(u') \Big] \Bigg) %\\	
		\end{split}
	\end{equation}
	
	Any $h \in \Real^\V$ that satisfies the above formula is a fixed point of $Av_1$. Hence, if the system has a solution where $h(v)>0.5$, the greatest fixed point $\ThO^{\reach(T)}$ satisfies $\ThO^{\reach(T)}(v)>0.5$. This is an instance of existential theory of reals which is known to be in PSPACE. Therefore, it is possible to decide $\ThO^{\reach(T)}(v)>0.5$ (equivalently $\ThT^{\reach(T)}(v)<0.5$) in PSPACE.
	
	\smallskip
	\noindent\textbf{Proof of (ii).}
	We provide the following reduction from the optimization problem to an instance of MILP; a different proof with a polynomial certificate is provided next. Let $O$ be the optimization problem as stated in the Section \ref{sec:reachability} with $\tau=0$.
	
	Let $M$ be any constant strictly greater than $\max_{u \in V}\{1+R_1(u)+R_2(u)\}$. For each node $u$ define two new variables $h^-(u), h^+(u)$ and add the following constraints to $O$:
	\[
	\begin{split}
		h^+(w) \geq&~ h(w)~ \forall w \in \N(u) \\
		h^+(w) \leq&~ h(w) + (1-b_u^w)\cdot M~ \forall w \in \N(u)\\
		\sum_{w \in \N(u)} &~b_u^w = 1 \\
		b_u^w \in \{0,1\} &~\forall w \in \N(u) \\
		\\
		h^-(w) \leq &~ h(w)~ \forall w \in \N(u) \\
		h^-(w) \geq &~ h(w) - (1-c_u^w)\cdot M~ \forall w \in \N(u)\\
		\sum_{w \in \N(u)} &~ c_u^w = 1 \\
		c_u^w \in \{0,1\} &~\forall w \in \N(u) \\		
	\end{split}
	\]
	This guarantees that $h^+(u) = h(u^+)$ and $h^-(u) = h(u^-)$, so they can be replaced. Next, replace each $\min(1,x)$ by $\frac{(x-1)- |x-1|}{2}+1$ and $\max(0,x)$ by $\frac{x+|x|}{2}$. Then replace each $|y|$ with a fresh variable $a_y$ and add the following constraints to $O$:
	\begin{enumerate}
		\item	$-a_y \leq y \leq a_y$
		\item	$y + M \cdot z_y \geq a_y \wedge -y + M \cdot (1-z_y) \geq a_y \wedge z_y \in \{0,1\}$
	\end{enumerate}
	The first constraint ensures that $|y| \leq a_y$ and the second one that $|y|\geq a_y$. Therefore, it is guaranteed that $|y|=a_y$. The MILP instance $O$ is equivalent to the optimization problem in section \ref{sec:reachability}. In order to decide whether $\ThO(v) \geq 0.5$ it suffices to decide satisfiability of $O$ with the additional constraint that $h(v)\geq 0.5$ and this decision problem is known to be in NP.
	
	\smallskip	
	\noindent\textbf{Alternative Proof for (ii).}
	A non-deterministic algorithm for deciding whether $\ThO^{\reach(T)}(v)>0.5$ can be characterized as follows: 
	\begin{enumerate}
		\item Guess $Z,O \subseteq V$ such that $Z \cap O = \varnothing$. Also for each $u \in V$, guess $u^+, u^-$ from successors of $u$. 
		\item Solve the following system of linear constraints:
		\[
		\begin{split}
			& h(v)\geq 0.5 \\
			&\forall u \in T,~  h(u)=0 \\
			&\forall u \in Z,~  h(u)=0 \wedge \frac{h(u^-)+h(u^+)}{2}(1+\RO(u)+\RT(u))-\RO(u) \leq 0 \\
			&\forall u \in O,~  h(u)=1 \wedge \frac{h(u^-)+h(u^+)}{2}(1+\RO(u)+\RT(u))-\RO(u) \geq 1 \\
			&\forall u \in V\setminus(T \cup Z \cup O)~,  0 \leq h(u) \leq 1 \\ 
			&~~~~~~~~~~~~~~~~~~~~~ \wedge  h(u) = \frac{h(u^-)+h(u^+)}{2}(1+\RO(u)+\RT(u))-\RO(u)\\
			& \forall u \in V,\forall u'\in \N(u),~ h(u^+)\geq h(u') \wedge h(u^-) \leq h(u')
		\end{split}
		\]
	\end{enumerate}
	The second step of the algorithm can be done in polynomial time. Hence, this is an NP algorithm for deciding whether $\ThO(v)\geq 0.5$ for specified $v \in V$. Dually, it is a coNP algorithm for deciding whether $\ThT(v)<0.5$. 
\end{proof}
%!TEX root=main.tex

\section{B\"uchi Bidding Games with Charging}
\label{sec:buchi}
We proceed to B\"uchi objectives. We point out that the proof of existence of thresholds is significantly more involved than for B\"uchi games with no charging. The key distinction is that thresholds in traditional strongly-connected B\"uchi games are trivial: If even one of the vertices is a B\"uchi target vertex, the B\"uchi player's threshold in \emph{each} vertex is $0$ and otherwise is $1$~\cite{AHC17}. This property gives rise to a simple reduction from traditional B\"uchi bidding games to reachability bidding games. With charging, this property no longer holds. For example, alter the game in Fig.~\ref{fig:examples:c} to make it strongly-connected by adding an edge from $t$ to $b$. The thresholds remain above $0$, i.e., there are initial budgets with which \PT wins. 

Our existence proof, which is inspired by an existence proof for discrete-bidding games~\cite{AS22}, follows a fixed-point characterization that is based on solutions to {\em frugal-reachability} games, which are defined below. We note that the proof has a conceptual similarity with Zielonka's algorithm~\cite{Zie98} in turn-based B\"uchi games, which characterizes the set of winning vertices based on repeated calls to an algorithm for turn-based reachability games.

\subsection{Frugal-Reachability Objectives}
We introduce frugal reachability objectives.
Consider a taxman-bidding game with charging $\G = \tup{\V,\E,\RO,\RT}$. Let $T\subseteq V$ be a set of target vertices and $\f: T \rightarrow [0,1]$ be a function that assigns each target with a \emph{frugal budget}. The \emph{frugal reachability} objective $\frugalreach(T,\f)$ requires \PO to reach $T$ such that the first time a vertex $v\in T$ is reached, \PO's budget must exceed $\f(v)$, thus:
\[
\frugalreach(T,\f) \coloneqq \set{ \tup{v_0,\BO^0}\tup{v_1,\BO^1}\ldots \mid  \exists i\;.\; v_i\in T\land \BO^i > \f(v_i) \land \forall j<i\;.\; v_j \notin T}
\]
We stress that $\frugalreach(T,\f)$ is a set of \emph{plays}, whereas the other objectives we have considered so far (reachability, B\"uchi, etc.) were sets of \emph{paths}.

Existence of thresholds $\ThO^{\frugalreach(T,\f)}$ and $\ThT^{\frugalreach(T,\f)}$ for the frugal-reachability objective and its dual are shown in Theorem \ref{thm:frugalreach}. For $v \in T$ and $t \in \Nat$, recall that we define $f_1(v, t) = 0$ (Def.~\ref{def:richman-reach-th}), which intuitively means that \PO wins if he reaches $v$ with {\em any} budget. Instead, we now define $f_1(v,t)=\f(v)$, requiring \PO to reach $v$ with a budget of $\f(v)$. Dually, we define $f_2(v,t)=1-\f(v)$.

\begin{definition}
	Define the functions $f_1,f_2: V \times \Nat \to [0,1]$ inductively:
	\[
	\begin{split}
		\underline{\forall v\in T}:&\quad \text{ for all } t\in \Nat,  f_1(v,t) \coloneqq \f(v) \text{ and } f_2(v, t) \coloneqq 1-\f(v)\\
		\underline{\forall v\notin T}:&\quad f_1(v,0) \coloneqq 1 \text{ and }  f_2(v,0) \coloneqq 0\\
		&\text{for } t > 0 \quad f_1(v,t) \coloneqq \updO\left( f_1(\cdot,t-1) \right)(v) \text{ and }\\
		&\qquad\qquad f_2(v,t) \coloneqq \updT\left(f_2(\cdot,t-1)\right)(v).
		%	    \rthO(v,t) &\coloneqq \S\left(\frac{\rthO(v^+,t-1)+\rthO(v^-,t-1)}{2}\big(1+R_1(v)+R_2(v)\big)- R_1(v)\right), \forall t>0,	    
	\end{split}
	\]
\end{definition}

\begin{restatable}{theorem}{ThmFrugReach}\label{thm:frugalreach}
The thresholds $\ThO^{\frugalreach(T,\f)}$ and $\ThT^{\frugalreach(T,\f)}$ exist. 
\end{restatable}

\begin{proof}
	\begin{enumerate}
		\item If \PO has more than $f_1(v,t)$ budget when the token is in $v$, he can win the $\frugalreach(T,\f)$ objective in at most $t$ steps. The proof is inductive and similar to the proof of Lem. \ref{lem:soundness of time-varying threshold}. The only difference is in the base of the induction, where for $t=0$ and $v \in T$, it is required that $f_1(v,t)=\f(v)$. 
		\item If \PT has more than $f_2(v,t)$ budget when the token is in $v$, she can prevent \PO from winning for at least $t$ steps. The proof of this statement is also similar to that of Lem. \ref{lem:soundness of time-varying threshold}. 
		\item For every $v,t$, it holds that $f_1(v,t)+f_2(v,t)=1$. 
		\item For every $v \in V$ and $t'>t$, it holds that $f_1(v,t') \leq f_1(v,t)$ and $f_2(v,t') \geq f_2(v,t)$. The proof is exactly the same as the of Lem. \ref{lem:monotonic}.
		\item The limit functions $f_1^*(v) = \lim_{t\to\infty} f_1(v,t)$ and $f_2^*(v) = \lim_{t\to\infty} f_2(v,t)$ are well-defined. Moreover, $f_1^*(v)$ is the greatest fixed point of $\updO$ subjected to constraints $f_1^*(v)=\f(v)$ for $v \in T$. 
		\item For every $v \in V$, it holds that $\ThO^{\frugalreach(T,\f)}(v) = f_1^*(v)$ and $\ThT^{\frugalreach(T,\f)}(v) = f_2^*(v)$.
	\end{enumerate}
\end{proof}

\subsection{Bounded-Visit B\"uchi and Co-B\"uchi}

We first prove the existence of thresholds for the simpler case of bounded-visit B\"uchi and co-B\"uchi objectives, where we impose, respectively, lower and upper bounds on the number of visits to the B\"uchi target vertices $B\subseteq \V$.
Let $k\in\mathbb{N}$ be a given bound.
The bounded-visit B\"uchi, denoted as $\buchi(B,k)$, intuitively requires \PO to visit $B$ at least $k$ times.
Formally, $\buchi(B,k)\coloneqq \set{v_0 v_1\ldots \mid \left|\set{i\in\mathbb{N}\mid v_i \in B}\right|\geq k}$.
Bounded-visit co-B\"uchi is the dual objective $\cobuchi(\compl{B},k)\coloneqq \set{v_0 v_1\ldots \mid \left|\set{i\in\mathbb{N}\mid v_i\in B}\right| < k} = \V^\omega\setminus \buchi(B,k)$.

Like before, we introduce two functions $g_1$ and $g_2$, which will be shown to characterize the thresholds for $\buchi(B,k)$ and $\cobuchi(B,k)$, respectively.

\begin{definition}
	\label{def:richman-buchi-th}
	Define the functions $\bthO,\bthT\colon \V\times \mathbb{N}\to [0,1]$ inductively as follows. 
	For every $v\in \V$, define $\bthO(v, 0) = 0$.
	For every $v \in B$, define $\bthO(v,1) \coloneqq 0$ and $\bthO(v,k) \coloneqq\updO(\bthO(\cdot,k-1))(v)$ for $k>1$, and for every $v \notin B$ and every $k > 0$, define $\bthO(v,k) \coloneqq \ThO^{\frugalreach(B,\bthO(\cdot,k))}(v)$. We proceed to define $\bthT$. 
	For every $v\in \V$, define $\bthT(v,0) \coloneqq 1$. For every $v\in B$, define $\bthT(v,1) \coloneqq 1$ and $\bthT(v,k) \coloneqq \updT(\bthT(\cdot,k-1))(v)$, for $k>1$, and for every $v \notin B$ and every $k > 0$, define $\bthT(v,k) \coloneqq \ThT^{\frugalreach(B,1-\bthT(\cdot,k))}(v)$.
\end{definition}

We prove the existence of thresholds and their correspondence to $g_1$ and $g_2$.

\begin{restatable}{lemma}{LemBuchiSoundTimeVaryingThresh}\label{lem:buchi:soundness of time-varying threshold}
For every $v \in V$ and $k \geq 0$, we have $g_1(v,k) = \Th_1^{\buchi(B, k)}(v)$ and $g_2(v,k) = \Th_2^{\buchi(B, k)}(v)$. Moreover, the thresholds exist: $\Th_1^{\buchi(B, k)}(v) + \Th_2^{\buchi(B, k)}(v) = 1$.
\end{restatable}
\begin{proof}
	We prove the statement using induction on $k$. For $k=0$ the statement holds trivially. Now suppose the token is in node $v$ and \PO has budget $B_1>g_1(v,k)$ for $k>1$: 
	\begin{itemize}
		\item If $v \in B$, then he can bid $b=\frac{g_1(v^+,k-1)-g_1(v^-,k-1)}{[g_1(v^+,k-1)-g_1(v^-,k-1)-1]\tau +2}$ and act as $\tup{v^-,b}$. Similar arguments as in the proof of Lem. \ref{lem:soundness of time-varying threshold} show that he will have budget more than $g_1(v^-,k-1)$ in case of winning and more than $g_1(v^+,k-1)$ in case of losing the bid. In any case, by induction hypothesis, he has a strategy to make the token visit $B$ at least $k-1$ more times. 
		\item If $v \notin B$, then his budget $B_1$ is strictly more than $\ThO^{\frugalreach(B,\f)}(v)$ where $\f: b \mapsto g_1(b,k)$. Therefore, he can use the strategy induced by the frugal reachability objective to ensure that $B$ is reached at least once, say in node $b \in B$, while his budget is more than $g_1(b,k)$. From the previous bullet-point it follows that he can then visit $B$ at least $k$ times from $b$. 
	\end{itemize}
	
	For \PT, suppose the token is in node $v$ and she has budget $B_2>g_2(v,k)$. Note that for $k=1$, the value of $g_2(v,k)$ coincides with $\ThT^{\reach(B)}(v)$. Therefore, as basis of the induction, if $B_2>g_2(v,1)$, she can prevent the token from ever reaching $B$. For $k>1$ there are two possible cases:
	\begin{itemize}
		\item If $v \in B$, she can bid $b=\frac{g_2(v^+,k-1)-g_2(v^-,k-1)}{[g_2(v^+,k-1)-g_2(v^-,k-1)-1]\tau +2}$ to ensure that she will have budget more than $g_2(u,k-1)$ at the next step. It follows by the induction hypothesis, that she can then limit the number of visits of $B$ to be less than $k-1$. 
		\item If $v \notin B$, she can play the \emph{frugal safety} strategy induced by $\ThT^{\frugalreach(B,\f)}$ where $\f: b \mapsto 1-g_2(b,k)$. It is therefore guaranteed that either (i) $B$ will not be reached at all, or (ii) If $B$ is reached at $b \in B$, then her budget will be at least $g_2(b,k)$, hence she can limit the number of visits of $B$ from then onward to be at most $k$. 
	\end{itemize}
	
	We also use induction on $k$ to show that $g_1(v,k)+g_2(v,k)=1$. For $k=0,1$ the statement holds trivially. For $k>1$, there are two cases:
	\begin{itemize}
		\item If $v \in B$, then similar to the proof of Lem. \ref{lem:soundness of time-varying threshold} it can be proved that $g_1(v,k)+g_2(v,k)=1$. 
		\item If $v \notin B$, then $g_1(v,k) = \ThO^{\frugalreach(B,\f_1)}(v)$ and $g_2(v,k) = \ThT^{\frugalreach(B,\f_2)}(v)$ where $\f_1: b \mapsto g_1(b,k)$ and $\f_2: b \mapsto 1-g_2(b,k)$. It follows from induction hypothesis that $\f_1=\f_2$. Then it follows from Thm. \ref{thm:frugalreach} that $g_1(b,k)+g_2(g,k)=1$.
	\end{itemize}
The proof then follows from the induction hypothesis. 
\end{proof}

We establish monotonicity of the thresholds, which is proved by the fact that $\PO$ needs higher budget for forcing larger numbers of visits to $B$.

\begin{lemma}\label{lem:bounded buchi operator:monotonicity}
For $v \in \V$ and $k \in \Nat$, we have $\Th_1^{\buchi(B, k)}(v) \leq \Th_1^{\buchi(B, k+1)}(v)$ and $\Th_2^{\buchi(B, k)}(v) \geq \Th_2^{\buchi(B, k+1)}(v)$. Moreover, the thresholds are bounded by $0$ and $1$. 
\end{lemma}

\subsection{Existence of Thresholds (for B\"uchi and Co-B\"uchi Objectives)}

We define two functions $\bThO$ and $\bThT$, which will be shown to coincide with the thresholds for the general (unbounded) B\"uchi and co-B\"uchi objectives, respectively.

\begin{definition}\label{def:buchi:threshold fixpoints}
	Define the functions $\bThO,\bThT\colon \V\to \mathbb{R}$ as follows. For every $v\in B$, define $\bThO(v)\coloneqq \lim_{k\to \infty} \bthO(v,k)$ and $\bThT(v) \coloneqq \lim_{k\to \infty} \bthT(v,k)$. For every $v\notin B$, define $\bThO(v) \coloneqq \ThO^{\frugalreach(B,\f)}(v)$ where $\f\colon b\mapsto \bThO(b)$ for every $b\in B$ and $\f\colon v\mapsto 0$ (can be arbitrary) for every $v\notin B$. 
	Likewise, for every $v\notin B$, define $\bThT(v) \coloneqq \ThT^{\frugalreach(B,\f)}(v)$ where $\f\colon b\mapsto 1-\bThT(b)$ for every $b\in B$ and $\f\colon v\mapsto 0$ (can be arbitrary) for every $v\notin B$. 
\end{definition}

Monotonicity (Lem.~\ref{lem:bounded buchi operator:monotonicity}) and boundedness of $g_1$ and $g_2$ imply the well-definedness of $g_1^*$ and $g_2^*$.
We now establish the existence and the characterization of thresholds.

\begin{restatable}{theorem}{thmBuchi}
For every $v \in V$, we have $\Th_1^{\buchi(B)}(v) = \bThO(v)$ and $\Th_2^{\buchi(B)}(v) = \bThT(v)$. Moreover, thresholds exist: $\Th_1^{\buchi(B)}(v) +\Th_2^{\buchi(B)}(v)=1$.
\end{restatable}

\begin{proof}
First, we show that $\bThO(v) \geq \Th_1^{\buchi(B)}(v)$.
Consider a configuration $\zug{v,\BO}$. When $\BO>g_1^*(v)$, \PO wins as follows. If $v\notin B$, he plays according to a winning strategy in a frugal-reachability game to guarantee reaching some $v' \in B$ with a budget that exceeds $g_1^*(v')$. For $v\in B$, he bids so that 
in the next configuration $\zug{v', B'_1}$, we have $B'_1 > g_1^*(v')$. 
Second, we show that $g_2^*(v)\geq \Th_2^{\buchi(B)}(v)$.
When $\BT = 1-B_1>g_2^*(v)$, \PT wins as follows. If $v \in B$, then there exists $k$ such that $\BT>g_2(v,k)$. Lem.~\ref{lem:buchi:soundness of time-varying threshold} shows that she can win the co-B\"uchi objective by preventing $B$ to be reached more than $k$ times. If $v \notin B$, she has a strategy to make the token either (i)~not reach $B$, or (ii)~reach $v' \in B$ with a budget at least $g_2^*(v')$. In both cases, she wins by repeating the strategy. 
Finally, by Lem.~\ref{lem:buchi:soundness of time-varying threshold}, we have 
$g_1(v, k) + g_2(v, k) = 1$, for all $k \in \Nat$. Thus, in the limit, we have $g_1^*(v) + g_2^*(v) = 1$, for $v \in B$. 
From this, the other sides of the above inequalities, i.e., $\bThO(v) \leq \Th_1^{\buchi(B)}(v)$ and $g_2^*(v)\leq \Th_2^{\buchi(B)}(v)$, and the existence claim follow in a straightforward manner.
\end{proof}

\subsection{Complexity Bounds (for B\"uchi and Co-B\"uchi Objectives)}

The computation of the thresholds $\ThO^{\buchi(B)} \equiv g_1^*$ and $\ThT^{\buchi(B)} \equiv g_2^*$ involves a nested fixed point computation.
For example, for $g_1^*$, the outer fixed point is the smallest fixed point of the sequence $g_1(\cdot,0),g_1(\cdot,1),\ldots$ for vertices in $B$, and for every $k=0,1,\ldots$, the inner fixed point is the usual greatest fixed point for frugal reachability thresholds required to reach $B$ with the leftover frugal budget $g_1(\cdot,k)$ from outside $B$.
The nested fixed point can be characterized as the solution of the following bilevel optimization problem.

\begin{align}\label{eq:optimization for richman buchi}
	\begin{split}
	&\qquad\min_{h\in \mathbb{R}^\V}\, \sum_{b\in B}h(b) \\
    &\text{subjected to constraints:}\\
    &\qquad h \in \arg\max_{h'\in\mathbb{R}^\V} \left\lbrace\sum_{v \in \V\setminus B} h'(v)\,\middle\vert\, \forall b\in B\;.\;h'(b)=h(b)\right\rbrace,\\
    &\qquad \forall v\in \V\;.\; h(v) = \updO\left(h(\cdot)\right)(v)\\
    &\qquad\qquad = \clamp{\frac{(1-\tau)h(v^-)+h(v^+)}{[h(v^+)-h(v^-)-1]\tau + 2}\cdot (1+\RO(v)+\RT(v)) - \RO(v)},\\
    &\qquad h(v^+) = \max_{u \in \N(v)} h(u), \quad
     h(v^-) = \min_{u \in \N(v)} h(u).
     \end{split}
\end{align}

\begin{restatable}{proposition}{PropRichmanBuchiSoundOpt} \label{prop:richman:buchi:sound opt}
	The solution of the optimization problem in \eqref{eq:optimization for richman buchi} is equivalent to the threshold function $\ThO^{\buchi(B)}$ of the B\"uchi player.
\end{restatable}

\begin{proof}%[Proof of Prop. \ref{prop:richman:buchi:sound opt}]
	Suppose we fix the values of $h(b)$ for each $b \in B$. Then, due to Thm. \ref{thm:frugalreach}, the value of $h(\cdot)$ for other vertices is the greatest fixed point of $\updO$. So, our goal is to find the least values for $h(b)$ values such that $\updO$ has at least one fixed point. Equivalently, we are looking for the least fixed point of $\updO$ with respect to $h(b)$ values and then the greatest fixed point with respect to other $h(\cdot)$ values. This is characterized in the bilevel optimization problem.
	
\end{proof}

\begin{restatable}{theorem}{ThmBuchiTaxmanPSPACE} \label{thm:buchi:taxman:pspace}
	The following hold:
	\begin{enumerate}[(i)]
		\item $\THRESH^{\mathrm{taxman}}_{\buchi}\in \textup{2EXPTIME}$ and $\THRESH^{\mathrm{taxman}}_{\cobuchi}\in \textup{2EXPTIME}$, and \label{complexity:taxman:buchi}
		\item $\THRESH^{\mathrm{Richman}}_{\buchi}\in \Pi_2^\mathrm{P}$ and $\THRESH^{\mathrm{Richman}}_{\cobuchi}\in \Sigma_2^\mathrm{P}$\label{complexity:richman:buchi}
	\end{enumerate}
\end{restatable}

	The bounds in \eqref{complexity:taxman:buchi} follow from a reduction to an equivalent query in the theory of reals.
	For \eqref{complexity:richman:buchi}, we can check if the solution of the optimization problem in \eqref{eq:optimization for richman buchi} is larger than $0.5$, and if this is true then we conclude that $0.5 < \ThO^{\buchi(B)}(v)$ and can output a \emph{negative} answer.
	Since \eqref{eq:optimization for richman buchi} is a bilevel MILP, hence the check can be done in $\Sigma_2^\mathrm{P}$ \cite{jeroslow1985polynomial}, and the overall complexity is $\Pi_2^\mathrm{P}$. The other case is dual. 
	
\begin{proof}%[Proof of Thm. \ref{thm:buchi:taxman:pspace}]
	\textbf{Proof of (i).} 
	\begin{comment}Consider the constraints in the proof of Thm. \ref{thm:reach:taxman:pspace} and denote them as $A^\G_T(h)$. Formally speaking given a bidding game $\G$, the formula $A^\G_T(h)$ is defined as follows:
		\[
		\begin{split}
			& \Big( \forall u \in \V\setminus T:\\
			&~~~~~~~~~~~~~~~~~\big[h(u)=0 \wedge \frac{(1-\tau)h(u^-)+h(u^+)}{[h(u^+) - h(u^-) -1] \tau +2}(1+\RO(u)+\RT(u))-\RO(u) \leq 0 \big] \\
			&~~~~~~~~~~~~~~~~~\vee  \big[h(u)=1 \wedge \frac{(1-\tau)h(u^-)+h(u^+)}{[h(u^+) - h(u^-) -1] \tau +2}(1+\RO(u)+\RT(u))-\RO(u) \geq 1 \big] \\
			&~~~~~~~~~~~~~~~~~\vee  \big [0 \leq h(u) \leq 1 \\ 
			& ~~~~~~~~~~~~~~~~~\wedge  h(u) = \frac{(1-\tau)h(u^-)+h(u^+)}{[h(u^+) - h(u^-) -1] \tau +2}(1+\RO(u)+\RT(u))-\RO(u) \big] \Big)\\
			& \wedge \Big( \forall u \in V, \Big[ \bigvee_{u' \in \N(u)} h(u^+)=h(u') \wedge \bigvee_{u' \in \N(u)} h(u^-)=h(u') \\
			& ~~~~~~~~~~\wedge \forall u'\in \N(u),~ h(u^+)\geq h(u') \wedge h(u^-) \leq h(u') \Big] \Big) \\	
		\end{split}
		\]
	\end{comment}
	The decision problem  $\THRESH^{\mathrm{taxman}}_{\buchi}$ is equivalent to satisfiability of the following quantified formula:
	
	\[
	\begin{split}
		\exists & h \in \Real^V \textit{ s.t. } h(v) \geq 0.5 \wedge A^\G_\varnothing(h)\\ 
		\wedge& \bigg( \forall h' \in \Real^V: [\big(\exists b \in B \textit{ s.t. } h'(b)<h(b)\big) \Rightarrow \neg A^\G_\varnothing(h') ] \\
		&\wedge \big[\big((\forall b \in B: h'(b)=h(b)) \wedge (\exists v \in V \setminus B \textit{ s.t. } h'(v)>h(v))\big) \\
		& \Rightarrow \neg A^\G_\varnothing(h')\big]
		\bigg)
	\end{split}
	\]
	Where $A^\G_\varnothing$ is defined in equation (\ref{eq:definition:AGT}). Intuitively, $h$ must be a fixed point of $Av_1$ where vertices in $B$ have the lowest possible values (minimization, least fixed point), and vertices not in $B$ have the greatest values (maximization, greatest fixed point).
	
	Quantifier elimination is known to be in 2EXPTIME. Therefore, deciding whether $h(v)\geq 0.5$ can be done in 2EXPTIME. 
	
	\smallskip
	\noindent\textbf{Proof of (ii).} Consider equation (\ref{eq:optimization for richman buchi}) and apply the same conversion as in the proof of Thm. \ref{thm:reach:taxman:pspace} to obtain an instance of bilevel mixed-integer linear programming which is equivalent to finding the thresholds.
\end{proof}

%!TEX root=main.tex

\section{Lower Complexity Bounds}
\label{sec:lower bounds}
In this section we show how to simulate a turn-based game using a Richman-bidding game with charging. Thus, solving Richman-bidding games with charging is at least as hard as their turn-based counterparts. Specifically, we obtain that solving Rabin bidding games with charging is NP-hard. This is a distinction from traditional Richman-bidding games, where solving Rabin games is in NP and coNP. 
Since taxman-bidding games generalize Richman-bidding games, hence it follows that Rabin taxman-bidding games are also NP-hard.

\paragraph*{Turn-based games}
A turn-based arena is a tuple $\tup{V,V_1,V_2,E}$, where
$V$ is a finite set of vertices,
$V_1,V_2$ form a partition of $V$, i.e., $V_1\cup V_2=V$ and $V_1\cap V_2=\emptyset$, and
$E\subseteq V\times V$ is a set of directed edges.
A turn-based $2$-player game is played on a turn-based arena for a given objective $\spec\subseteq V^\omega$.
Similar to bidding games, the game is played by \PO and \PT, moving a token along the edges of the arena.
But unlike bidding games, who moves the token is decided based on whether it is placed on a vertex in $V_1$ or $V_2$, respectively.
A history is a finite sequence $v^0v^1\ldots v^n\in V^*$, and strategies of \PO and \PT are functions $\rho_1\colon V^*V_1\to V$ and $\rho_2\colon V^*V_2\to V$, with $\rho_i(hv)\in \N(v)$ for $i\in\set{1,2}$.
A play starting at an initial vertex $v$ is an infinite sequence of vertices, starting at $v$, and is obtained by applying any given pair of strategies, written as $\play(v,\rho_1,\rho_2)$.
\PO can win the game from a given vertex $v$, if he has a strategy $\rho_1$ such that for every \PT strategy $\rho_2$, we have $\play(v,\rho_1,\rho_2)\in \spec$.

\begin{restatable}{lemma}{LemRedTurnBased}\label{lem:reduction from turn-based game}
	Given a turn-based game $\G$, an initial vertex $v$, and an objective $\spec$, there is a bidding game with charging $\G'$ of size linear in $\G$, with the same objective and initial vertex such that \PO can win $\G$ from $v$ if and only if $\tup{\G,v,\spec}\in \THRESH_S^{\mathrm{Richman}}$.
\end{restatable}
\begin{proof}
%The definitions of turn-based games and the detailed proof can be found in App.~\ref{app:proof of lower bounds}. 
	The reduction is described below.
Suppose $\G = \tup{V,V_1,V_2,E}$ is the turn-based arena.
Define $\G' = \tup{V',E',\RO,\RT}$ such that
$\V' \coloneqq \V\cup \set{s_1,s_2}$ with $s_1,s_2$ being two new sink vertices,
$\E' \coloneqq \E \cup V_1\times \set{s_1} \cup V_2\times \set{s_2}\cup  \set{(s_1,s_1),(s_2,s_2)}$,
for every $v\in V_1$, $\RO(v) \coloneqq 2$ and $\RT(v)\coloneqq 0$, 
for every $v\in V_2$, $\RO(v) \coloneqq 0$ and $\RT(v) \coloneqq 2$, and
$\RO(s_1) \coloneqq \RO(s_2) \coloneqq \RT(s_1) \coloneqq\RT(s_2)\coloneqq 0$.
The initial vertex $v'\coloneqq v$.
Finally, $\spec' \coloneqq \set{\xi = v^0v^1\ldots \in {\V'}^\omega \mid \xi\in \spec \cap \safe(\set{s_1})} \cup \reach(\set{s_2})$.

Intuitively, $\G'$ contains the same set of vertices as $\G$ with two additional sink vertices $s_1$ and $s_2$, where $s_i$ is losing for \PLi, for $i \in \set{1,2}$. For every vertex $v$, if $v$ is controlled by \PO in $\G$, then in $\G'$, we define \PO's charge to be $\RO(v) = 2$. Moreover, we add an edge from $v$ to $s_1$, requiring \PO to win the bidding in $v$. Note that even if \PO starts with a budget of $\epsilon > 0$, at $v$, after charging and normalization, his budget exceeds $2/3$. \PT's vertices are dual. It is not hard to verify that \PO can win in $\G$ from $v$ if and only if $\Th^{\spec}(v)=0$, and \PT can win in $\G$ from $v$ if and only if $\Th^{\spec}(v)=1$. 
\end{proof}

Since turn-based Rabin games are NP-hard, we obtain the following.

\begin{theorem}\label{thm:rabin and streett lower bounds}
We have $\THRESH^{\mathrm{Richman}}_{\rabin}\in \textup{NP-hard}$ and $\THRESH^{\mathrm{Richman}}_{\streett}\in \textup{coNP-hard}$. 
\end{theorem}

%!TEX root=main.tex

\section{Repairing Bidding Games} \label{sec:repair}

\newcommand{\repaired}{\mathit{Repaired}}
\newcommand{\REPAIR}{\mathrm{R\_THRESH}}

\begin{wrapfigure}{r}{0.65\textwidth}
	\tikzset{every state/.style={minimum size=0pt}}
	\begin{tikzpicture}[node distance=0.5cm]
		\node[state,initial,label={right:$\begin{bmatrix}0\\ 2\end{bmatrix}$}]	(a)	at	(0,0)		{$a$};
		\node[state]			(b)	[above left=of a]		{$b$};
		\node[state]			(c)	[left=of b]			{$c$};
		\node[state]			(d)	[above right=of a]		{$d$};
		\node[state]			(e)	[right=of d]			{$e$};
		\node[state]			(f)	[above=of b]			{$f$};
		\node[state,accepting]			(g)	[above=of d]			{$g$};
		
		\path[->]	(a)		edge	(b)
		edge	(d)
		(b)		edge[bend left]		(c)
		edge	(f)
		edge	(g)
		(c)		edge[bend left]	(b)
		(d)		edge[bend left]		(e)
		edge	(f)
		edge	(g)
		(e)		edge[bend left]		(d);
	\end{tikzpicture}
	\begin{tabular}[b]{c c}
		$\RO'$	&	$\ThO(a)$\\
		\hline
		$-$		&	$1$\\
		$c,e\mapsto 1$	&	$1$\\
		$b\mapsto 2$	&	$0.75$\\
		$d\mapsto 2$	&	$0.75$\\
		$a,\ldots,e\mapsto 0.4$	&	$0.62$\\
		$a\mapsto 2$	&	$0.5$\\
		$b,d\mapsto 1$		&	$0$
	\end{tabular}
	\caption{Illustrating the repair problem.
		LEFT: A reachability game with the objective $\reach(\set{g})$.
		RIGHT: With no repair (first row), $\ThO(a) = 1$. We depict repairs (first col.) and the changes they imply to the thresholds (second col.), for a repair budget of $C = 2$.
	}
	\label{fig:reinforcement}
	\vspace{-0.5cm}
\end{wrapfigure}

In this section, we introduce the {\em repair} problem for bidding games. Intuitively, the goal is to add minimal charges to the vertices of an arena so as to decrease the threshold in the initial vertex to a target threshold. Formally, we define the following problem.

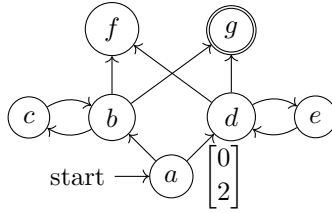
\begin{definition}[Repairing bidding games]
Consider an arena $\G$ with a vertex $v$, a bidding mechanism $M\in \set{\mathit{Richman},\mathit{poorman},$ $\mathit{taxman}}$, a class of objectives $S\in \set{\reach,\safe,\buchi,\cobuchi}$, and a repair budget $C\in \mathbb{R}_{\geq 0}$. The set of \emph{repaired arenas}, denoted $\repaired(\G,C)$, are arenas obtained from $\G$ by adding \PO charges whose sum does not exceed $C$.
Formally, $\repaired(\G,C) \coloneqq \set{\tup{\V,\E,\RO',\RT} \text{ is an arena}  \mid \forall v\in\V\;.\; \RO'(v)\geq \RO(v) \land \sum_{v\in\V} (\RO'(v) - \RO(v)) \leq C}$. The problem $\REPAIR^M_S$ takes as input $\zug{\G, v, \varphi, C}$, where $\varphi \in S$, and accepts iff there exists $\G'\in\repaired(\G,C)$ with $\tup{\G',v,\varphi}\in \THRESH^M_S$.
\end{definition}

\begin{example}
We illustrate the non-triviality of the repair problem in Fig.~\ref{fig:reinforcement}.
Observe that neither assigning charges uniformly nor assigning charges to a single vertex, decrease the threshold sufficiently, whereas adding a charge of $1$ to both $b$ and $d$ is a successful repair. \qed
\end{example}

\begin{theorem} 
\label{thm:repair:Richman}
	The following hold:
	\begin{enumerate}[(i)]
		\item $\REPAIR^{\mathrm{Richman}}_{\reach}\in \textup{2EXPTIME}$, \label{thm:repair:Richman:reachability}
		\item $\REPAIR^{\mathrm{Richman}}_{\safe}\in \textup{PSPACE}$. \label{thm:repair:Richman:safety}
	\end{enumerate}
\end{theorem}

\begin{proof} 
	We introduce notation for the proof.
	Let $G = \tup{\V,\E,\RO,\RT}$ be a bidding game and $U$ be a set of vertices. 
	Define $A^G_U\colon [0,1]^V\to \set{0,1}$ such that for every $h\in [0,1]^V$, $A^G_U(h)=1$ iff $h(v) = \updO(h)(v)$ for every $v\notin U$.
	Observe that $\Th_1^{\reach(T)}$ is the largest $h$ for which $A^G_T(h)=1$ and moreover $h(v) = 0$ for every $v\in T$.

	\noindent\textbf{Proof of \eqref{thm:repair:Richman:reachability}:}
	Consider a bidding game $G = \tup{\V,\E,\RO,\RT}$ where the objective of \PO is $\reach(T)$ for some $T \subseteq V$. The goal is to check if it is possible to increase $\RO$ by a total of $C$ such that the reachability threshold at $a \in \V$ falls below $0.5$. This is equivalent to: 
	\begin{multline*}
	\exists \RO' \in \Real^\V \;.\; 
	\left(\RO' \geq \RO\right) 
	\wedge 
	\left(|\RO'-\RO|_1 \leq C\right) 
	\wedge \\
	\left(\forall \ThO \in \Real^\V\;.\; \Big[\big(\forall v \in T\;.\; \ThO(v)=0 \big) \wedge A^{\G'}_T(\ThO)\Big] \Rightarrow \ThO(a) \leq 0.5 \right)
	\end{multline*}
	where $\G'=\tup{\V,\E,\RO',\RT}$. 
	The validity of the above formula can be checked by applying a quantifier elimination method. Therefore, the decision problem is in $\textup{2EXPTIME}$.
	
	\noindent\textbf{Proof of \eqref{thm:repair:Richman:safety}:}
	Consider a bidding game $G = \tup{\V,\E,\RO,\RT}$ where the objective of \PO is $\safe(T)$ for some $T \subseteq V$. The goal is to check if it is possible to increase $\RO$ by a total of $C$ such that the safety threshold at $a \in \V$ falls below $0.5$. This is equivalent to:
	\begin{multline*}
	\exists \RO' \in \Real^\V \;.\; 
	\left(\RO' \geq \RO\right)
	 \wedge 
	 \left(|\RO'-\RO|_1 \leq C \right)
\wedge \\
	\left(\exists \ThO \in \Real^\V \textit{ s.t. } \big[\ThO(a) \leq 0.5 \wedge A^{\G'}_T(\ThO) \wedge \forall u \in T, \ThO(u)=1 \big]\right)
	\end{multline*}
	where $\G'=\tup{\V,\E,\RO',\RT}$.
	The above formula can be seen as an input instance of existential theory of reals which is known to be in $\textup{PSPACE}$.
\end{proof}

\section{Conclusion and Future Work}
We introduce and study a generalization of bidding games in which players' budgets are charged throughout the game.
One advantage of the model over traditional bidding games is that long-run safety is not trivial. We show that the model maintains the key favorable property of traditional bidding games, namely the existence of thresholds, 
the proof of which is, however, significantly more challenging due to the non-uniqueness of thresholds.
We characterize thresholds in terms of greatest and least fixed points of certain monotonic operators. 
Finally, we establish the first complexity lower bounds in continuous-bidding games and study, for the first time, a repair problem in this model. 

There are plenty of open questions and directions for future research. 
First, it is important to extend the results to richer classes of $\omega$-regular objectives, like parity, Rabin, and Streett, as well as to quantitative objectives, like mean-payoff. 
Second, tightening the complexity bounds is an important open question. For example, it might be the case that finding thresholds in Richman-bidding games with charging is in NP and coNP. 
Third, traditional reachability Richman-bidding games are equivalent to a class of stochastic games~\cite{Con92} called {\em random-turn games}~\cite{PSSW09}, and the equivalence is general and intricate in infinite-duration games~\cite{AHC19,AHI18,AHZ21,AJZ21}. It is unknown if such a connection exists for games with charging, and if it does, then many of the open questions may be solved via available tools for stochastic games.
Finally, there are various possible extensions, like charges disappearing after a vertex is visited, charges that are collectible in multiple installments, etc.

\bibstyle{plainurl}
\bibliography{references,ga}

\end{document}